\documentclass[a4paper,12pt]{article}

\usepackage{fullpage}
\usepackage{amsmath}
\usepackage{amsthm}
\usepackage{amsfonts}
\usepackage{mathrsfs}
\usepackage{xspace}

\title{{\bf Definable Inapproximability: \\ New Challenges for
    Duplicator}}

\author{Albert Atserias \\
  Departament de Ci\`encies de la Computaci\'o \\
  Universitat Polit\`{e}cnica de Catalunya 
  \and
  Anuj Dawar \\
  Department of Computer Science and Technology \\ 
  University of Cambridge}

\newcommand{\logic}[1]{\ensuremath\mathrm{#1}\xspace}
\newcommand{\FO}{\logic{FO}}
\newcommand{\FOC}{\logic{FOC}}
\newcommand{\FP}{\logic{FP}}
\newcommand{\FPC}{\logic{FPC}}

\newcommand{\cclass}[1]{\ensuremath\mathrm{#1}\xspace}
\newcommand{\NP}{\cclass{NP}\xspace}
\newcommand{\PT}{\cclass{P}\xspace}

\newcommand{\prob}[1]{\ensuremath\textsc{#1}\xspace}
\newcommand{\vc}[1]{\mathrm{vc}(#1)}
\newcommand{\vcd}[1]{\mathrm{vcd}(#1)}
\newcommand{\is}[1]{\mathrm{IS}(#1)}
\newcommand{\his}[2]{\mathrm{IS}_{#1}(#2)}
\newcommand{\isd}[1]{\mathrm{isd}(#1)}

\newcommand{\FF}[1]{\ensuremath{\mathbb{F}_{#1}}\xspace}

\newcommand{\NN}{\ensuremath\mathbb{N}\xspace}
\newcommand{\Ck}{\ensuremath{C^k}}
\newcommand{\Lk}{\ensuremath{L^k}}
\newcommand{\ELk}{\ensuremath{\exists L^{k,+}}}

\newcommand{\Cequiv}[1]{\ensuremath{\equiv_{C^{#1}}}}

\newcommand{\ELsub}[1]{\ensuremath{\Rightarrow_{#1}}}

\newcommand{\struct}[1]{\ensuremath\mathbb{#1}}
\newcommand{\Aa}{\struct{A}}
\newcommand{\Bb}{\struct{B}}

\newcommand{\class}[1]{\ensuremath\mathscr{#1}}

\newcommand{\Exi}{\ensuremath{\exists^i}}

\newcommand{\ra}{\rightarrow}

\newcommand{\powerset}[1]{\mathscr{P}(#1)}
\newcommand{\kpowerset}[2]{\mathscr{P}_{#1}(#2)}
\newcommand{\nats}{\mathbb{N}}
\newcommand{\vect}[1]{\overline{#1}}

\newtheorem{lemma}{Lemma}
\newtheorem{theorem}[lemma]{Theorem}

\newtheorem{corollary}[lemma]{Corollary}
\newtheorem{claim}[lemma]{Claim}

\date{}

\begin{document}

\maketitle

\begin{abstract}
We consider the hardness of approximation of optimization problems
from the point of view of definability.  For many~$\NP$-hard
optimization problems it is known that, unless~$\PT = \NP$, no
polynomial-time algorithm can give an approximate solution guaranteed
to be within a fixed constant factor of the optimum.  We show, in
several such instances and without any complexity theoretic
assumption, that no algorithm that is expressible in fixed-point logic
with counting (FPC) can compute an approximate solution.  Since
important algorithmic techniques for approximation algorithms (such as
linear or semidefinite programming) are expressible in FPC, this
yields lower bounds on what can be achieved by such methods.  The
results are established by showing lower bounds on the number of
variables required in first-order logic with counting to separate
instances with a high optimum from those with a low optimum for
fixed-size instances.
 \end{abstract}

\section{Introduction}

Twenty years ago, the PCP theorem~\cite{pcp} transformed the landscape
of complexity theory.  It showed that if~$\PT \neq \NP$ then not only
is it impossible to efficiently solve~$\NP$-hard problems exactly but
for some of them it is also impossible to approximate the solution to
within a constant factor.  Consider for instance the
problem~$\prob{MAX 3SAT}$.  Here we are given a Boolean formula
in~$\prob{3CNF}$ and we are asked to determine~$m^*$, the maximum number
of clauses 
that can be simultaneously satisfied by an assignment of Boolean
values to its variables.  It is a consequence of the PCP theorem that
there is a constant~$c < 1$ such that, assuming~$\PT \neq \NP$, no
polynomial-time algorithm can be guaranteed to produce an assignment
that satisfies at least~$cm^*$ clauses, or indeed determine the value
of~$m^*$ up to a factor of~$c$.  The proof of the PCP theorem introduced
sophisticated new techniques into complexity theory such as the
\emph{probabilistically checkable proofs} that gave the theorem its
name.  Over the years, stronger results were proved, improving the
constant~$c$ and, by reductions, proving inapproximability results for
a host of other~$\NP$-hard problems.

A structural theory of hardness of approximation was introduced by
Papadimitriou and Yannakakis~\cite{PY} who defined the
class~$\prob{MAX SNP}$ of approximation problems, with a definition
rooted in descriptive complexity theory.  They showed that for every
problem
in this class, there is a constant~$d$ and a
polynomial-time algorithm can find approximate solutions within a
factor~$d$ of the optimum.
At the
same time, for all problems 
that are~$\prob{MAX SNP}$-hard, under approximation-preserving 
reductions defined in~\cite{PY},
there is a constant~$c$ such that no polynomial-time algorithm can
approximate solutions 
within a factor~$c$.  This makes it a challenge, for each~$\prob{MAX
  SNP}$-complete problem, to determine the exact approximation ratio
that is achievable by an efficient algorithm.  In some cases, this has
been pinned down exactly.  For instance, for~$\prob{MAX 3SAT}$ we know
that there is a polynomial-time algorithm that will produce an
assignment satisfying~$7/8$ of the clauses in any formula but,
unless~$\PT = \NP$, there is no polynomial-time algorithm that is
guaranteed to produce a solution within~$7/8+ \epsilon$ of the
optimal, for any~$\epsilon > 0$~\cite{Hastad2001}.  Another
interesting case is~$\prob{MAX 3XOR}$, where we are given a formula
which is the conjunction of clauses, each of which is the~$\prob{XOR}$
of three literals.  Here, satisfiability is decidable in polynomial
time as the problem is essentially that of solving a system of linear
equations over the two-element field.  However, determining, for an
unsatisfiable system,
how many of its clauses can be simultaneously satisfied is~$\prob{MAX
  SNP}$-hard, and the exact approximation ratio that is achievable
efficiently is known: unless~$\PT = \NP$, no polynomial-time algorithm
can achieve an approximation ratio bounded
above~$1/2$~\cite{Hastad2001}.

To give a problem of another flavour, consider \emph{minimum vertex
  cover}, the problem of finding, in a graph~$G$, a minimum set~$S$ of
vertices such that every edge is incident on a vertex in~$S$.
Let~$\vc{G}$ denote the size of a minimum size vertex cover in~$G$.
There are algorithms that are guaranteed to find a vertex cover no
larger than~$2\vc{G}$ (this being a minimization problem, the
approximation ratio is expressed as a number~$c \geq 1$).  It has been
proved, by means of rather sophisticated reductions starting at the
PCP theorem, that, unless~$\PT = \NP$, no polynomial-time algorithm
can achieve a ratio better than~$1.36$~\cite{DinurSafra}.  Very recent
results announced in~\cite{KhotMinzerSafra} improve this lower bound
to~$\sqrt{2}$.  It is conjectured that indeed no such algorithm could
achieve a ratio of~$2 - \epsilon$ for arbitrarily small~$\epsilon > 0$
but, as of our current knowledge, the right threshold constant could
be somewhere between~$\sqrt{2}$ and~$2$.

We approach these questions on the hardness of approximability from
the point of view of definability.  Our aim is to show that the tools
of descriptive complexity can be brought to bear in showing lower
bounds on the definability of approximations and that these
definability lower bounds have consequences on understanding commonly
used techniques in approximation algorithms.

A reference logic in descriptive complexity is fixed-point logic with
counting,~$\FPC$.  The class of problems definable in this logic form
a proper subclass of the complexity class~$\PT$.  However,~$\FPC$ is
very expressive and many natural problems in~$\PT$ are expressible in
this logic.  For instance, any polynomial-time decidable problem on a
proper-minor closed class of graphs is expressible
in~$\FPC$~\cite{Grohe-book}.  Also, problems that can be formulated as
linear programming or semidefinite programming problems are
in~$\FPC$~\cite{ADH15,AO18,DW-lics17}.  At the same time, for many
problems we are able to prove categorically, i.e., without complexity
theoretic assumptions, that they are not definable in~$\FPC$.  Among
these are~$\NP$-complete problems like~$\prob{3SAT}$,
graph~$3$-colourability and Hamiltonicity (see~\cite{DawarSigLog}).
We can also prove that certain problems in~$\PT$ are not in~$\FPC$,
such as~$\prob{3XOR}$~\cite{ABD09,CFI92}.

A particularly interesting class of problems is the class of optimization
problems known as~$\prob{MAX CSP}$ or constraint maximization
problems, where we are given a collection of constraints and the
problem is to find the maximum number of constraints that can be
simultaneously satisfied.  When it comes to finding exact solutions,
definability in~$\FPC$ turns out to be an excellent guide to the
tractability of such problems.  It is known that each such problem is
either in~$\PT$ \emph{and} definable in~$\FPC$ or it is~$\NP$-complete
and provably \emph{not} definable in~$\FPC$~\cite{DW15}.  We would
like to extend such results also to the \emph{approximability} of such
problems.  This paper develops the methodology for doing so.

For~$\prob{MAX 3SAT}$, we prove, without any complexity theoretic
assumption, that no algorithm expressible in~$\FPC$ can achieve an
approximation ratio of~$7/8 + \epsilon$.  The question seems ill-posed
at first sight as~$\FPC$ is a formalism for defining problems rather
than expressing algorithms.  We return to the precise formulation
shortly, but first note that there is a sense in which~$\FPC$ can
express, say the ellipsoid method for solving linear
programs~\cite{ADH15}.  This is the basis for showing that many
commonly used algorithmic techniques for approximation problems, such
as semidefinite programming relaxations, are also expressible
in~$\FPC$.  Thus, on the one hand, reductions from~$\prob{MAX
  SNP}$-hard problems show inapproximability by \emph{any}
polynomial-time algorithm, assuming~$\PT \neq \NP$.  On the other
hand, our results show, without the assumption, inapproximability by
the most commonly used polynomial-time methods.

Undefinability of a class of structures~$\class{C}$ in~$\FPC$ is
typically established by showing that structures in~$\class{C}$ cannot
be distinguished from structures not in~$\class{C}$
in~$\Ck$---first-order logic with counting and just~$k$
variables---for any fixed~$k$.  In the terminology
of~\cite{DW-lics17},~$\class{C}$ has unbounded \emph{counting width}.
On the other hand, hardness of approximation for a maximization
problem
is typically established by showing
that every class that includes all instances 
with an optimum~$m^*$ and excludes all instances with an optimum less
than~$cm^*$, is~$\NP$-hard.  Our method combines these two.  We aim to
show that any class separating instances 
with an optimum~$m^*$ from instances with an optimum less than~$cm^*$ has
unbounded counting width.  In general, we not only show that counting
width is unbounded, but establish stronger bounds on how it grows with
the size of instances, as such bounds are directly tied to lower
bounds on semidefinite programming hierarchies~\cite{DW-lics17,AO18}.  This
methodology poses new challenges for Spoiler-Duplicator games in
finite model theory.  Such games are typically played on pairs of
structures that are minimally different. In the new setting, we need
to show Duplicator winning strategies in games on pairs of structures
that differ substantially, on some numeric~parameters.

The PCP theorem is the \emph{fons et origo} of results on hardness of
approximation.  It established the first provably~$\NP$-hard constant
gap between the fully satisfiable instances of~$\prob{MAX 3SAT}$,
i.e., those in which all clauses can be satisfied, and the less
satisfiable ones, those where no more than~$1-\epsilon_0$ can be
satisfied, for some explicit~$\epsilon_0 > 0$. The gap between~$1$
and~$1-\epsilon_0$ was then amplified and also transferred to other
problems by means of reductions.  For us, the starting point is the
problem~$\prob{MAX 3XOR}$.  We are able to establish a definability
gap between the satisfiable instances of this and instances in which
little more than~$1/2$ of the clauses can be satisified.  The
constant~$1/2$ is easily seen to be optimal since in
every~$\prob{3XOR}$ instance at least half of the equations can be
satisfied.

The methods for establishing this optimal \emph{initial gap} are very
different from that for the PCP theorem. We construct a~$k$-locally
satisfiable instance of~$\prob{MAX 3XOR}$ which, by a random
construction is at the same time highly unsatisfiable.  We can then
combine this with a construction adapted from~\cite{ABD09} to obtain
an optimal gap that defeats any fixed counting width.  This shows that
no algorithm that is expressible in~$\FPC$ can approximate~$\prob{MAX
  3XOR}$ within a constant above~$1/2$, even on satisfiable instances.
It should be pointed out that, although the inapproximability
of~$\prob{MAX 3XOR}$ above~$1/2$ matches algorithmic lower bounds and
is tight, the \emph{type} of definability gap that we obtain, which
applies to satisfiable instances, \emph{cannot} have an analogue in
the algorithmic setting.  The satisfiable instances of~$\prob{MAX
  3XOR}$ \emph{are} distinguished from unsatisfiable ones by a
polynomial-time algorithm.  To show inapproximability for any constant
greater than~$1/2$ one has to show that it is the almost satisifable
ones that are indistinguishable from those that are highly
unsatisfiable.  This distinction supports our claim that our methods
are very different from that for the PCP theorem.

With such an optimal initial gap for~$\prob{MAX 3XOR}$ in hand, we can
then transfer it to other problems by means of reductions, just as in
classical inapproximability.  Our reductions have to preserve~$\FPC$
definability and we mostly rely on first-order definable reductions.
For one, the standard direct reduction from~$\prob{3XOR}$
to~$\prob{3SAT}$ is trivially first-order definable and gives an
optimal undefinability gap for~$\prob{MAX 3SAT}$: no algorithm
expressible in~$\FPC$ can achieve an approximation ratio
of~$7/8+\epsilon$, even on satisfiable instances. Again this matches
known algorithm lower bounds and is tight. For other problems we need
to rely on more sophisticated constructions, without leaving the realm
of first-order definable reductions. It turns out that many of the
reductions used in the classical theory of approximability \emph{are}
first-order reductions but this requires close examination and proof.

We show that the \emph{long-code reductions} from \cite{Hastad2001}
are definable in first-order logic. Such reductions have the merit of
providing different constructions of optimal gaps for~$\prob{MAX
  3XOR}$ and~$\prob{MAX 3SAT}$ starting at \emph{any} initial gap
whatsoever. In addition, the techniques that are involved in them have
applications elsewhere. For the vertex cover problem, we are able to
show that the reduction from~\cite{DinurSafra}, which is based on the
same long-code reduction techniques as in~\cite{Hastad2001}, is
first-order definable, showing that~$\FPC$ cannot give an
approximation better than~$1.36$.  It is possible that this could be
improved to~$\sqrt{2}$ using the recent breakthrough
of~\cite{KhotMinzerSafra} but we leave this to future work.

\section{Preliminaries} \label{sec:preliminaries}

We use~$\FF{2}$ to denote the~$2$-element field.  For any positive
integer~$n$, let~$[n] := \{1,\ldots,n\}$.

\subsection{Logics and games}  

We assume familiarity with first-order logic~$\FO$.  All our
vocabularies are finite and relational, and all structures are finite.
For a structure~$\struct{A}$, we write~$A$ to denote its universe, and
we often write~$|\struct{A}|$ and~$|A|$ interchangeably to mean the
number of elements in the universe.  We refer to fixed-point
logic~$\FP$ and fixed-point logic with counting~$\FPC$ but the
definitions of these are not required for the technical development in
the paper.  For this it suffices to consider the bounded variable
fragments of first-order logic.

For a fixed positive integer~$k$, we write~$\Lk$ to denote the
fragment of first-order logic in which every formula has at most~$k$
variables, free or bound.  We also write~$\ELk$ for the
\emph{existential positive} fragment of~$\Lk$.  This consists of those
formulas of~$\Lk$ formed using only the positive Boolean
connectives~$\land$ and~$\lor$, and existential
quantification. ~$\FOC$ is the extension of first-order logic with
\emph{counting quantifiers}.  For each natural number~$i$, we have a
quantifier~$\Exi$ where~$\struct A \models \Exi x \,\phi$ if, and only
if, there are at least~$i$ distinct elements~$a \in A$ such
that~$\struct A\models \phi[a/x]$. While the extension of first-order
logic with counting quantifiers is no more expressive than~$\FO$
itself, the presence of these quantifiers does affect the number of
variables that are necessary to express a query.  Let~$\Ck$ denote
the~$k$-variable fragment of~$\FOC$ in which no more than~$k$
variables appear, free or bound.

For two structures~$\struct A$ and~$\struct B$, we write~$\struct A
\Cequiv{k} \struct B$ to denote that they are not distinguished by any
sentence of~$\Ck$. All that we need to know about~$\FPC$ is that for
every formula~$\phi$ of~$\FPC$ there is a~$k$ such that if~$\struct{A}
\Cequiv{k} \struct{B}$ then~$\struct{A} \models \phi$ if, and only
if,~$\struct{B} \models \phi$.  We also write~$\struct A \ELsub{k}
\struct B$ to denote that every sentence of~$\ELk$ that is true
in~$\struct A$ is also true in~$\struct B$.  While~$\Cequiv{k}$ is an
equivalence relation,~$\ELsub{k}$ is reflexive and transitive but not
symmetric. These relations have well established characterizations in
terms of two-player pebble games.  The relation~$\ELsub{k}$ is
characterized by the \emph{existential~$k$-pebble game}~\cite{KV90}
and~$\Cequiv{k}$ by the \emph{$k$-pebble bijective game}~\cite{Hel96}.

Both versions of the game are played on a pair of structures~$\struct
A$ and~$\struct B$ by two players, Spoiler and Duplicator, using~$k$
pairs of pebbles~$(a_1,b_1),\dots,(a_k,b_k)$.  In a game position,
some (or all) of the pebbles~$a_1,\ldots,a_k$ are placed on elements
of~$\struct A$ while the matching pebbles among~$b_1,\ldots,b_k$ are
placed on elements of~$\struct B$.  Where it causes no confusion, we
do not distinguish notationally between the pebble~$a_i$ (or~$b_i$)
and the element on which it is placed.  In the
\emph{existential~$k$-pebble game}, at each move Spoiler chooses a
pebble~$a_i$ (which might or might not already be on an element
of~$\struct A$) and places it on any element of~$\struct A$.
Duplicator has to respond by placing~$b_i$ on an element of~$\struct
B$.  If the resulting partial map from~$\struct A$ to~$\struct B$
given by~$a_i \mapsto b_i$ is not a partial homomorphism, then Spoiler
has won the game.  In the \emph{$k$-pebble bijective game} Spoiler
chooses a pair of pebbles~$(a_i,b_i)$ and Duplicator has to respond by
giving a bijection~$f: \struct{A} \ra \struct{B}$ which agrees with
the map~$a_j \mapsto b_j$ for all~$j \neq i$.  Spoiler chooses a
pair~$(a,f(a))$ on which to place the pebbles~$(a_i,b_i)$.  Again, if
the resulting partial map from~$\struct A$ to~$\struct B$ given
by~$a_i \mapsto b_i$ is not a partial isomorphism, then Spoiler has
won the game.  In both games, we say Duplicator has a \emph{winning
  strategy} if, no matter how Spoiler plays, it can play forever
without losing.  The following summarises the connection between these
games and the relations~$\Cequiv{k}$, and~$\ELsub{k}$.  For any two
structures~$\struct{A}$ and~$\struct{B}$, the following
hold:~$\struct{A} \ELsub{k} \struct{B}$ if, and only if, Duplicator
has a winning strategy in the existential~$k$-pebble game played
on~$\struct{A}$ and~$\struct{B}$~\cite{KV90}; and~$\struct{A}
\Cequiv{k} \struct{B}$ if, and only if, Duplicator has a winning
strategy in the~$k$-pebble bijective game played on~$\struct{A}$
and~$\struct{B}$~\cite{Hel96}.

For undirected graphs, the relation~$\Cequiv{2}$ has a simple
combinatorial characterization in terms of \emph{vertex refinement}
(see~\cite{IL90}).  For any graph~$G$, there is a coarsest
partition~$C_1,\ldots,C_m$ of the vertices of~$G$ such that for
each~$1 \leq i,j \leq m$ there exists~$\delta_{ij}$ such that each~$v
\in C_i$ has exactly~$\delta_{ij}$ neighbours in~$C_j$.  Let~$H$ be
another graph and~$D_1,\ldots D_{m'}$ be the corresponding partition
of~$H$ with constants~$\gamma_{ij}$.  Then~$G \Cequiv{2} H$ if, and
only if,~$m = m'$ and there is a permutation~$h \in \mathrm{Sym}_m$
such that~$|C_i| = |D_{h(i)}|$ and~$\delta_{ij} = \gamma_{h(i)h(j)}$
for all~$i$ and~$j$.

All classes of structures we consider in this paper are assumed to be
closed under isomorphism.  Let~$\class{C}$ be such a class of
structures and for any~$n \in \NN$, let~$\class{C}_n$ denote the
structures in~$\class{C}$ with at most~$n$ elements.  The
\emph{counting width} of~$\class{C}$~\cite{DW-lics17} is the
function~$k: \NN \ra \NN$ where~$k(n)$ is the smallest value such that
for any~$\struct{A} \in \class{C}_n$ and any~$\struct{B} \not\in
\class{C}$, we have~$\struct{A} \not\Cequiv{k(n)} \struct{B}$.  Note
that~$k(n) \leq n$.  Because~$\struct{A} \not\Cequiv{1} \struct{B}$
whenever~$\struct{A}$ and~$\struct{B}$ have different numbers of
elements,~$k(n)$ is also the smallest value such that~$\class{C}_n$ is
a union of~$\Cequiv{k(n)}$-classes.  In particular, it follows that
the counting width of~$\class{C}$ is the same as that of its
complement.  For~$k: \NN \ra \NN$, we say that two disjoint
classes~$\class{C}$ and~$\class{D}$ are \emph{$\Ck$-separable} if
whenever~$\struct{A} \in \class{C}_n$ and~$\struct{B} \in
\class{D}_n$, then we have~$\struct{A} \not\Cequiv{k(n)} \struct{B}$.
Equivalently~$\class{C}$ and~$\class{D}$ are \emph{$\Ck$-separable} if
there is a class~$\class{E}$ of counting width ~$k$ such
that~$\class{C} \subseteq \class{E}$ and~${\class{D}} \subseteq
\overline{\class{E}}$.  To see that the two conditions are equivalent,
first suppose that whenever~$\struct{A} \in \class{C}_n$
and~$\struct{B} \in \class{D}_n$ we have~$\struct{A} \not\Cequiv{k(n)}
\struct{B}$.  Then, if we define~$\class{E}$ to be the set that
contains, for every~$n$, all structures of size~$n$ that
are~$\Cequiv{k(n)}$-equivalent to some structure in~$\class{A}$, it
witnesses the second condition.  In the other direction, if for
some~$n$, we have~$\struct{A} \in \class{C}_n$ and~$\struct{B} \in
\class{D}_n$ and~$\struct{A} \Cequiv{k(n)} \struct{B}$, then
any~$\class{E}$ with counting width~$k$ that contains~$\struct{A}$
must also contain~$\struct{B}$.

\subsection{Interpretations}  

Consider two vocabularies~$\sigma$ and~$\tau$.  A \emph{$d$-ary
  ~$\FO$-interpretation of~$\tau$ in~$\sigma$} is a sequence of
first-order formulas in vocabulary~$\sigma$ consisting of: (i) a
formula~$\delta(\vect{x})$; (ii) a formula~$\varepsilon(\vect{x},
\vect{y})$; (iii) for each relation symbol~$R \in \tau$ of arity~$k$,
a formula~$\phi_R(\vect{x}_1, \dots, \vect{x}_k)$; and (iv) for each
constant symbol~$c \in \tau$, a formula~$\gamma_c(\vect{x})$, where
each~$\vect{x}$,~$\vect{y}$ or~$\vect{x}_i$ is a~$d$-tuple of
variables.  We call~$d$ the \emph{dimension} of the interpretation.
If~$d = 1$, we say that the interpretaion is \emph{linear}.  We say
that an interpretation~$\Theta$ associates a~$\tau$-structure~$\Bb$ to
a~$\sigma$-structure~$\Aa$ if there is a map~$h$ from~$\{ \vect{a} \in
A^d \mid \Aa \models \delta[\vect{a}] \}$ to the universe~$B$ of~$\Bb$
such that: (i)~$h$ is surjective onto~$B$; (ii)~$h(\vect{a}_1) =
h(\vect{a}_2)$ if, and only if,~$\Aa\models \varepsilon[\vect{a}_1,
  \vect{a}_2]$; (iii)~$R^{\Bb}(h(\vect{a}_1), \dots, h(\vect{a}_k))$
if, and only if,~$\Aa \models \phi_R[\vect{a}_1, \dots, \vect{a}_k]$;
and (iv)~$h(\vect{a}) = c^{\Bb}$ if, and only if,~$\Aa \models
\gamma_c[\vect{a}]$.  Note that an interpretation~$\Theta$ associates
a~$\tau$-structure with~$\Aa$ only if~$\varepsilon$ defines an
equivalence relation on~$A^d$ that is a congruence with respect to the
relations defined by the formulae~$\phi_R$ and~$\gamma_c$. In such
cases, however,~$\Bb$ is uniquely defined up to isomorphism and we
write~$\Theta(\Aa) = \Bb$.  It is also worth noting that the size
of~$\Bb$ is at most~$n^d$, if~$\Aa$ is of size~$n$.  But, it may in
fact be smaller.  We call an interpretation~$p$-bounded, for a
polynomial~$p$, if~$|\Bb| \leq p(|\Aa|)$, and say the interpretation
is \emph{linearly bounded} if~$p$ is linear.  Every linear
interpretation is linearly bounded, but the converse is not
necessarily the case.

For a class of structures~$\class{C}$ and an interpretation~$\Theta$,
we write~$\Theta(\class{C})$ to denote the class~$\{\Theta(\struct{A})
\mid \struct{A} \in \class{C}\}$.  We mainly use interpretations to
define reductions between classes of structures.  These allow us to
transfer bounds on separability, by the following lemma.

\begin{lemma}\label{lem:reduction}
  Let~$\Theta$ be a~$p$-bounded interpretation of dimension~$d$ and
  let~$t$ be the maximum number of variables appearing in any formula
  of~$\Theta$.  If~$\class{C}$ and~$\class{D}$ are two disjoint
  classes of structures such that~$\Theta(\class{C})$
  and~$\Theta(\class{D})$ are~$C^{k(n)}$-separable, then~$\class{C}$
  and~$\class{D}$ are ~$C^{dk(p(n)) + t}$-separable.
\end{lemma}
\begin{proof}
  Let~$\Aa \in \class{C}_n$ and~$\Bb \in \class{D}_n$ be two
  structures.  Then, since~$\Theta(\Aa)$ and~$\Theta(\Bb)$ have size
  at most~$p(n)$, there is a formula~$\phi \in C^{k(p(n))}$ such
  that~$\Theta(\Aa) \models \phi$ and~$\Theta(\Bb) \not\models \phi$.
  We compose~$\phi$ with the interpretation~$\Theta$ to
  obtain~$\phi'$.  That is to say, we replace every relation symbol by
  its defining formula, including replacing all occurrences of
  equality by~$\varepsilon$, and we relativize all quantifiers
  to~$\delta$.  Note that this involves replacing quantification over
  elements with quantification over tuples.  That is to say, we need
  assertions of the form ``there exist~$i$ tuples~$\vect{x}$ such that
  \ldots''.  It is well known that such counting quantifiers over
  tuples can be replaced by a series of counting quantifiers over
  single elements without increasing the total number of variables.
  Then~$\Aa \models \phi'$ and ~$\Bb \not\models \phi'$.  It is also
  easy to check that~$\phi'$ has at most~$dk(p(n)) + t$ variables.
  The multiplicative factor~$d$ comes from the fact that every
  variable in~$\phi$ is replaced by a ~$d$-tuple and the additive~$t$
  accounts for any other variables that may appear in the formulas
  of~$\Theta$.
\end{proof}

When we wish to define a reduction from a class ~$\class{C}$ by a
first-order interpretation, it suffices to give an
interpretation~$\Theta$ for all structures in~$\class{C}$ with at
least two elements (or, indeed, at least~$k$ elements for any
fixed~$k$).  This is because we can define an arbitrary map on a
finite set of structures by a first-order formula, so we just need to
take the disjunction of~$\Theta$ with the formula that defines the
required interpretation on the structures with one element.  With this
in mind, we define the method of \emph{finite expansions} which gives
us interpretations~$\Theta$ that take a structure~$\Aa$ with
universe~$A$ to a structure with a universe consisting of~$l$ labelled
disjoint copies of~$S$ for some definable subset~$S$ of~$A$.  Note
that~$\Theta$ would not, in general, be linear, but it is linearly
bounded.

So, fix a value~$l$, and let~$t$ be the least integer such that~$l
\leq 2^t$.  In a structure~$\Aa$ with at least two elements, we say
that a~$t+1$-tuple of elements~$(a_1,\ldots,a_{t+1})$ \emph{codes} an
integer~$i \in [2^t]$ if~$b_1\cdots b_t$ is the binary representation
of~$i-1$ and for all~$j\in [t]$ we have~$b_j = 1$ if, and only
if,~$a_{j+1} \neq a_{1}$.  For each~$i$, we can clearly define a
formula~$\gamma_i(\vect{y})$ with~$t+1$ free variables that defines
those tuples that code~$i$.  Now, for any formula~$\phi(x)$,
let~$\delta(x,\vect{y})$ be the formula~$\phi(x) \land \bigvee_{i \leq
  l} \gamma_i(\vect{y})$ and let~$\epsilon(x_1,\vect{y}_1,
x_2,\vect{y}_2)$ be the formula ~$$x_1 = x_2 \land \bigvee_{i}
\gamma_i(\vect{y}_1) \land \gamma_i(\vect{y}_2).$$ In other
words,~$\delta$ picks out those~$t+2$ tuples~$(s,\vect{a})$ where~$s$
satisfies~$\phi$ and~$\vect{a}$ codes an integer in~$[l]$,
and~$\epsilon$ identifies distinct tuples which have the same~$s$ and
the same integer~$l$.  An interpretation using these can be seen to
yield a structure with~$l$ disjoint copies of the set of elements
of~$\Aa$ satisfying~$\phi$.

\section{The Basic Gap Construction}

The problems~$\prob{3SAT}$ and~$\prob{3XOR}$ both ask to decide if a
formula consisting of the conjunction of Boolean constraints each on
exactly three Boolean variables is satisfiable. In~$\prob{3SAT}$ the
constraints are disjunctions of literals on three distinct
variables. In~$\prob{3XOR}$ the constraints are parities of three
distinct variables. Both problems are known to have unbounded counting
width~\cite{ABD09}: the class of satisfiable instances cannot be
separated in~$\Ck$, for bounded~$k$, from the class of unsatisfiable
ones. Our aim is to show that this result can be strengthened to show
that the class of satisfiable instances is not~$\Ck$-separable from
the class of instances that are \emph{highly unsatisfiable}, meaning
that no assignment to the variables can satisfy more than a
fraction~$s$ of the constraints for some fixed~$s \in (0,1)$.  We give
a basic construction for~$\prob{3XOR}$, based on that in~\cite{ABD09},
that establishes this for any~$s > 1/2$, with a lower bound on the
value of~$k$ that is linear in the number of variables in the
system. Then we use this construction to get one for~$\prob{3SAT}$ for
any~$s > 7/8$, also for a value of~$k$ that is linear in the number of
variables. In both cases, the constants~$1/2$ and~$7/8$ are known to
be optimal.

\subsection{Systems of constraints} \label{sec:encodings}

Let~$\Gamma$ be a finite set of relations over a finite domain~$D$,
also called a \emph{constraint language}. Let~$I = \{ c_1,\ldots,c_m
\}$ be a collection (multi-set) of constraints, each of the
form~$R(x_{i_1},\ldots,x_{i_k})$, where~$R$ is a~$k$-ary relation
in~$\Gamma$, and~$x_{i_1},\ldots,x_{i_k}$ are~$k$ distinct~$D$-valued
variables from a set~$x_1,\ldots,x_n$ of~$n$ variables. For~$c \in
[0,1]$, we say that the system~$I$ is~$c$-satisfiable if there is an
assignment~$f : \{x_1,\ldots,x_n\} \rightarrow D$ that satisfies at
least~$cm$ constraints; i.e., that
satisfies~$(f(x_{i_1}),\ldots,f(x_{i_k})) \in R$ for at least~$cm$
constraints~$R(x_{i_1},\ldots,x_{i_k})$ from~$I$.  Note that, as we
are counting the number of satisfied constraints, multiplicities
matter and this is why we have multi-sets rather than sets of
constraints.

We think of a system~$I = \{ c_1,\ldots,c_m \}$ over the constraint
language~$\Gamma$ as a finite structure in two ways. In the first
encoding, the universe is the disjoint union of~$x_1,\ldots,x_n$
and~$c_1,\ldots,c_m$. The vocabulary includes binary
relations~$E_1,E_2,\ldots$ such that~$E_i(x,c)$ holds if the
constraint~$c$ has arity~$i$ or more and~$x$ is the~$i$th variable
in~$c$. The vocabulary also includes a unary relation~$Z_R$ for each
relation~$R$ in~$\Gamma$ such that~$Z_R(c)$ holds if~$c$ is
an~$R$-constraint: a constraint of the
form~$R(x_{i_1},\ldots,x_{i_k})$ for some
variables~$x_{i_1},\ldots,x_{i_k}$, where~$k$ is the arity of~$R$. In
the second encoding, the universe is just the set of
variables~$x_1,\ldots,x_n$, and the vocabulary includes a~$k$-ary
relation symbol~$R$ for each~$k$-ary relation~$R$ in~$\Gamma$, such
that~$R(x_{i_1},\ldots,x_{i_k})$ holds if this is one of the
constraints in the collection~$c_1,\ldots,c_m$. Note that in this
second encoding the collection of constraints is treated as a set. In
particular, the multiplicity of constraints is lost, which could
affect its~$c$-satisfiability.

The constraint language~$\Gamma$ is also encoded as a finite structure
in two ways. In the first encoding the domain is~$D^{\leq r} = D \cup
D^2 \cup D^3 \cup \cdots \cup D^r$, where~$r$ is the maximal arity of
a relation in~$\Gamma$. The relations~$E_1,E_2,\ldots$ are interpreted
by the projections:~$E_i(b,(b_1,\ldots,b_k))$ holds for~$b \in D$
and~$(b_1,\ldots,b_k) \in D^k$ if, and only if,~$i \leq k$ and~$b =
b_i$. The relations~$Z_R$ are interpreted by the relation~$R$ itself
as a unary relation over the universe:~$Z_R((b_1,\ldots,b_k))$ holds
if~$k$ is the arity of~$R$ and~$(b_1,\ldots,b_k)$ belongs to~$R$. In
the second encoding, the universe is just~$D$, and the relation
symbol~$R$ is interpreted by~$R$ itself.  Where it causes no
confusion, we do not distinguish between a constraint
language~$\Gamma$ and the structure that encodes it, and similarly
between an instance~$I$ and its encoding structure.

It is easily seen that, in both encodings as finite structures, a
system~$I$ over~$\Gamma$ is satisfiable if, and only if, there is a
homomorphism from the structure that encodes~$I$ to the structure that
encodes~$\Gamma$.  We say that the system is \emph{$k$-locally
  satisfiable} if~$I \ELsub{k} \Gamma$.

For~$\prob{3SAT}$, the constraint language is
denoted~$\Gamma_{\prob{3SAT}}$. It has domain~$D = \{0,1\}$ and the
relations are the eight relations~$R_1,\ldots,R_8 \subseteq \{0,1\}^3$
defined by the eight possible clauses on three variables.
For~$\prob{3XOR}$, the constraint language is
denoted~$\Gamma_{\prob{3XOR}}$. It also has domain~$D = \{0,1\}$ and
the relations are the two relations~$R_0,R_1 \subseteq \{0,1\}^3$
defined by the two possible linear equations~$x + y + z = b$ with
three variables over~$\FF{2} = \{0,1\}$. Accordingly,~$\prob{3XOR}$
instances~$I$ can be identified with systems of linear
equations~$Ax=b$ over~$\FF{2}$. In the following,~$A$ and~$b$ are
referred to as the left-hand side matrix of~$I$ and right-hand side
vector of~$I$, respectively.

It is probably useful to spell out, in simple words, what it means for
a~$\prob{3SAT}$ or~$\prob{3XOR}$ instance to be~$k$-locally
satisfiable. Intuitively, what this means is that every set of less
than~$k$ variables induces a satisfiable subformula and that, in
addition, at least one of the satisfying assignments that exists
can be extended to any other variable, still satisfying the resulting
induced subformula, and itself satisfying the same type of extension
property. Strictly speaking this description is accurate only in the
\emph{second} encoding; in the first encoding one has to consider
sets of less than~$k$ variables \emph{and} clauses, and then the
correspondence with satisfying assignments with the extension property
is not as direct.  However, what is true and useful (and easy to see)
is that if a~$\prob{3SAT}$ or~$\prob{3XOR}$ instance is~$k$-locally
satisfiable in the first encoding then it is also~$k$-locally satisfiable in
the second encoding, while if it is~$3k$-locally satisfiable in the
second encoding, then it is also~$k$-locally satisfiable in the first
encoding.  We use the second of these claims in the proof of Lemma~\ref{lem:linear-local-bound} below
(where we also spell out the easy proof of it).

\subsection{Gap construction}

We now focus on~$\prob{3XOR}$ and hence on systems of linear equations
over~$\FF{2}$.

A starting point for us is the following construction which allows us
to convert any \emph{$k$-locally satisfiable} system of equations into
a pair of systems that are~$\Cequiv{k}$-indistinguishable.
See~\cite[Prop.~32]{ADW-lics17} for a related construction, which is
inspired by the proof in~\cite{ABD09} that satisfiability of systems
of linear equations over~$\FF{2}$ is not invariant under~$\Cequiv{k}$
for any~$k$.

For any instance~$I$ of~$\prob{3XOR}$ we define another
instance~$G(I)$ of~$\prob{3XOR}$ which has two variables~$x_j^0$
and~$x_j^1$ for each variable~$x_j$ of~$I$.  For each
equation~$x_j+x_k+x_l = b$ in~$I$, we have eight equations in~$G(I)$
given by the eight possible values of~$a_1,a_2,a_3 \in \{0,1\}$
in~$x_j^{a_1}+x_k^{a_2}+x_l^{a_3} = b + a_1 + a_2 + a_3$.  If~$I$ is
the system~$Ax=b$, then the homogeneous companion of~$I$ is the
system~$Ax=0$, which we denote~$I^0$. Note that the system~$G(I^0)$ is
satisfiable for any~$I$ by setting each variable~$x_j^a$ to~$a$. We
show that, despite this, as long as~$I$ is locally satisfiable,~$G(I)$
is hard to distinguish from its homogeneous companion~$G(I^0)$.

\begin{lemma}\label{lem:CFI}
  For every~$\prob{3XOR}$ instance~$I$ and every integer~$k \geq 3$,
  if~$I$ is~$k$-locally satisfiable, then~$G(I) \Cequiv{k} G(I^0)$.
\end{lemma}
\begin{proof}
We describe a strategy for Duplicator in the~$k$-pebble bijective game
played on~$G(I)$ and~$G(I^0)$, given a strategy in the
existential~$k$-pebble game on~$I$ and~$\Gamma =
\Gamma_{\prob{3XOR}}$.

Suppose we have a position in the existential~$k$-pebble game on~$I$
and~$\Gamma$ with pebbles on~$x_1,\ldots,x_{k'}$, for some~$k' \leq k$
in~$I$, and corresponding pebbles on~$v_1,\ldots,v_{k'} \in \{0,1\}$
in~$\Gamma$.  Suppose further that this is a winning position for
Duplicator, i.e.\ she has a strategy to play forever from this
position.  Then, we claim that the position in the bijective game
where the pebbles in~$G(I)$ are on~$x_1^{a_1},\ldots,x_{k'}^{a_{k'}}$,
for some~$a_1,\ldots,a_{k'} \in \{0,1\}$ and the matching pebbles
in~$G(I^0)$ are on~$x_1^{a_1+v_1},\ldots,x_{k'}^{a_{k'}+v_{k'}}$ is a
winning position in the bijective game on these two structures.  To
see this, note first that, if~$x_r+x_s+x_t=b_i$ is an equation in~$I$,
for~$1 \leq r,s,t \leq k'$, then by assumption that the position is
winning in the existential game,~$v_r+v_s+v_t = b_i$.
Hence,~$x_r^{a_r}+x_s^{a_s}+x_t^{a_t} = b_i$ is an equation in~$G(I)$
if, and only if,~$x_r^{a_r}+x_s^{a_s}+x_t^{a_t} = 0$ is an equation
in~$G(I^0)$ if, and only
if,~$x_r^{a_r+v_r}+x_s^{a_s+v_s}+x_t^{a_t+v_t} = v_r+v_s+v_t$ is an
equation in~$G(I^0)$, but this last equation
is~$x_r^{a_r+v_r}+x_s^{a_s+v_s}+x_t^{a_t+v_t} = b_i$.  Thus, the map
from~$x_1^{a_1},\ldots,x_{k'}^{a_{k'}}$
to~$x_1^{a_1+v_1},\ldots,x_{k'}^{a_{k'}+v_{k'}}$ is a partial
isomorphism.  To see that Duplicator can maintain the condition,
suppose Spoiler moves the pebbles on~$(x^{a_j}_j,x^{a_j+v_j}_j)$.  By
assumption, Duplicator has a response in the existential game whenever
Spoiler moves the pebble from~$x_j$ to~$x_l$.  This response defines a
function~$f$ from the variables in~$x$ to~$\{0,1\}$.  We use this to
define the bijection taking~$x_l^a$ to~$x_l^{a+f(x_l)}$.  This is a
winning move in the bijective game.
\end{proof}

As far as the degree of satisfiability is concerned, the construction
preserves a gap in the following quantifiable terms:

\begin{lemma}\label{lem:both-lift}
For every~$\prob{3XOR}$ instance~$I$ and every~$c,s \in [0,1]$, the
following hold:
\begin{enumerate}\itemsep=0pt 
\item if~$I$ is~$c$-satisfiable, then~$G(I)$ is~$c$-satisfiable,
\item if~$I$ is not~$s$-satisfiable, then~$G(I)$ is
  not~$(1/2+s/2)$-satisfiable.
\end{enumerate}
\end{lemma}
\begin{proof}
For proving 1, let~$h: \{x_1,\ldots,x_n\} \rightarrow \{0,1\}$ be an
assignment of values to the variables of~$I$ that satisfies at
least~$cm$ of the~$m$ equations in~$I$.  Define the assignment~$g$ on
the variables of~$G(I)$ by~$g(x^a) = h(x) + a$.  For each equation~$e$
satisfied by~$h$, all eight equations arising from~$e$ are satisfied
by~$g$ and so~$g$ satisfies at least~$8cm$ of the~$8m$ equations
in~$G(I)$.
  
For proving 2, suppose~$g$ is an assignment of values in~$\{0,1\}$ to
the variables~$x^a_i$ in~$G(I)$.  Let~$h: \{x_1,\ldots,x_n\} \ra
\{0,1\}$ be the assignment defined by~$h(x_j) = g(x^0_j)$. We claim
that if~$e_i$ is an equation~$x_j + x_k + x_l = b$ in~$I$ that is not
satisfied by~$h$ then at least four of the eight equations in ~$G(I)$
arising from~$e_i$ are falsified by~$g$.  To see this, consider two
cases.  First, suppose that~$g(x_t^0) = g(x_t^1)$ for some~$t \in
\{j,k,l\}$.  Without loss of generality, we assume~$t = j$.  Then
consider the four \emph{pairs} of equations
  \begin{align*}
  & x_j^0 + x_j^{a_1} + x_k^{a_2} = b_i + a_1 + a_2, \\
  & x_j^1 + x_j^{a_1} + x_k^{a_2} = b_i + a_1 + a_2 + 1 
  \end{align*}
obtained by taking the four possible values of~$a_1$ and~$a_2$.
Since~$g(x_j^0) = g(x_j^1)$, if one equation in a pair is satisfied
by~$g$ the other is necessarily falsified.  Thus, at least four
equations are falsified.  For the second case, suppose that for
each~$t \in \{j,k,l\}$ occurring in~$e_i$ we have~$g(x_t^0) \neq
g(x_t^1)$.  But then, since we assume that~$h$ falsifies~$e_i$, it
follows that~$g$ falsifies~$x_j^0+ x_k^0+x_l^0 = b$ and hence it
falsifies all eight equations arising from~$e_i$.  In either case,~$g$
falsifies at least four of the equations arising from~$e_i$.

Now, suppose that~$g$ satisifes at least~$(1/2+s/2)\cdot 8m$ of
the~$8m$ equations in~$G(I)$.  We claim that~$h$ satisfies at
least~$sm$ equations in~$I$.  Suppose for contradiction that~$h$
falsifies a proportion~$\epsilon > 1 - s$ of the equations.  By the
above argument, then~$g$ falsifies at least~$4\epsilon m$ of the
equations in~$G(I)$.  But~$4 \epsilon m > (1/2-s/2)\cdot 8m$
contradicting the assumption that~$g$ satisfies at least~$(1/2+s/2)
\cdot 8 m$ equations.
\end{proof}

The extreme cases of Lemma~\ref{lem:both-lift} are given by~$c = 1$
for point~\emph{1}, and~$s = 1/2+\epsilon$ with~$\epsilon > 0$ for
point~\emph{2}.  Indeed, every~$\prob{3XOR}$ instance
is~$1/2$-satisfiable, as witnessed by the all-zero assignment, or the
all-one assignment, whichever satisfies more equations. Note also that
point~\emph{1} of Lemma~\ref{lem:both-lift} preserves its extremality:
if~$I$ is satisfiable, then so is~$G(I)$. However, point~\emph{2} does
not preserve its extremality, since even if~$I$ is
not~$(1/2+\epsilon)$-satisfiable, the best that can be claimed
about~$G(I)$ is that it is not~$(3/4 + \epsilon/2)$-satisfiable.  In
the following we show that if the vector~$b$ is chosen uniformly at
random, then both instances~$Ax=b$ and~$G(Ax=b)$ are at
most~$(1/2+\epsilon)$-satisfiable, with high probability, provided the
matrix~$A$ has at least a constant-factor more rows than columns.

\begin{lemma} \label{lem:both-extremal}
For every two reals~$\epsilon > 0$ and~$\delta > 0$ there exists an
integer~$c > 0$ such that for every sufficiently large integer~$n$ and
every matrix~$A \in \{0,1\}^{m \times n}$, where~$m = cn$ and each row
of~$A$ has exactly three ones, if~$b$ is chosen uniformly at random
in~$\{0,1\}^m$ then, with probability at least~$1-\delta$,
both~$\prob{3XOR}$ instances~$Ax=b$ and~$G(Ax=b)$ are at
most~$(1/2+\epsilon)$-satisfiable.
\end{lemma}

\begin{proof} 
Fix~$\epsilon > 0$ and~$\delta > 0$ and let~$c$ be any integer bigger
than~$1/\epsilon^2$. Let~$n$ be sufficiently large, let~$m = cn$, and
let~$A \in \{0,1\}^{U \times V}$ be any matrix with~$U = [m]$ and~$V =
[n]$ that has exactly three ones in each row. For each ~$b = (b_u)_{u
  \in U} \in \{0,1\}^U$, the instance~$Ax=b$ has one variable~$x_v$
for each~$v \in V$ and one equation~$e_u : x_{v_1(u)} + x_{v_2(u)} +
x_{v_3(u)} = b_u$ for each~$u \in U$, where ~$v_1(u),v_2(u),v_3(u) \in
V$ are the three columns of~$A$ that have ones in row~$u$.  The
instance~$G(Ax=b)$ has three variables~$x_v^a$ for each~$v \in V$ and
eight equations~$e_u^{a_1,a_2,a_3}$ for each ~$u \in U$.

For each assignment~$f : \{ x_v : v \in V \} \rightarrow \{0,1\}$ for
the variables of~$Ax=b$ and each~$u \in U$, let~$X_{f,u}$ be the
indicator random variable for the event that~$f(x_{v_1(u)}) +
f(x_{v_2(u)}) + f(x_{v_3(u)}) = b_u$; i.e., for the event that~$f$
satisfies the equation~$x_{v_1(u)} + x_{v_2(u)} + x_{v_3(u)} =
b_u$. The probability of this event is~$1/2$, and all such events,
as~$u$ ranges over~$U$, are mutually independent. Thus, setting~$X_f =
\sum_{u \in U} X_{f,u}$, we have that~$X_f$ is a binomial random
variable with expectation~$\mathbb{E}[X_f] = m/2$. By Hoeffding's
inequality, the probability that~$X_f - \mathbb{E}[X_f] \geq t$ is at
most~$e^{-2t^2/m}$. In particular, the probability that~$X_f \geq
(1/2+\epsilon)m$ is at most~$e^{-2\epsilon^2 m}$.  By the union bound,
the probability that some~$f$ satisfies~$X_f \geq (1/2+\epsilon)m$ is
at most~$2^n e^{-2\epsilon^2 m}$.
  
Similarly, for each assignment~$f : \{ x_v^a : v \in V,\; a \in
\{0,1\} \} \rightarrow \{0,1\}$ for the variables of~$G(Ax=b)$ and
each~$u \in U$, let~$Y_{f,u} \in [0,1]$ be the fraction of equations
of~$G(Ax=b)$ among those that come from~$u$ that are satisfied by~$f$;
i.e., precisely,~$Y_{f,u}$ is~$1/8$-th of the number of
triples~$(a_1,a_2,a_3) \in \{0,1\}^3$ for which the
equality~$f(x_{v_1(u)}^{a_1}) + f(x_{v_2(u)}^{a_2}) +
f(x_{v_3(u)}^{a_3}) = b_u + a_1 + a_2 + a_3$ holds. We claim that the
expectation of the random variable~$Y_{f,u}$ is~$1/2$. To see this,
consider two cases: 1)~$f(x_{v_j(u)}^{0}) \not= f(x_{v_j(u)}^{1})$ for
all~$j \in \{1,2,3\}$, and 2)~$f(x_{v_j(u)}^{0}) = f(x_{v_j(u)}^{1})$
for some~$j \in \{1,2,3\}$. In case 1), either all eight equations
that come from~$u$ are satisfied, or none is, and each possibility
happens with probability~$1/2$ according to the outcome of the random
choice of~$b_u$. The expectation of~$Y_{f,u}$ is thus~$1/2$ in this
case. In case 2), exactly half of the eight equations that come from
~$u$ are satisfied, and which half depends on the outcome of the
random choice of~$b_u$. The expectation of~$Y_{f,u}$ is thus~$1/2$
also in this case. This shows that the expectation of~$Y_{f,u}$ is
~$1/2$ in either case. Moreover, the random variables~$Y_{f,u}$, as
~$u$ ranges over~$U$, are mutually independent. Thus, setting~$Y_f =
(1/m)\sum_{u \in U} Y_{f,u}$, we have that~$Y_f$ is the average of
~$m$ independent random variables with range in~$[0,1]$. By
Hoeffding's inequality, the probability that~$Y_f - \mathbb{E}[Y_f]
\geq t$ is at most~$e^{-2t^2 m}$. In particular, the probability that
~$Y_f \geq 1/2+\epsilon$ is at most~$e^{-2\epsilon^2 m}$. By the union
bound, the probability that some~$f$ satisfies~$Y_f \geq 1/2+\epsilon$
is at most~$2^n e^{-2\epsilon^2 m}$.

Since~$m = cn$ and~$c > 1/\epsilon^2$, twice~$2^n e^{-2\epsilon^2 m}$
is at most~$2^{n+1} e^{-2n}$ and is less than~$\delta$ for all
sufficiently large values of~$n$. Thus, for any large enough~$n$, the
probability that both~$Ax=b$ and~$G(Ax=b)$ are at
most~$(1/2+\epsilon)$-satisfiable is at least~$1-\delta$.
\end{proof}

The next step is to show that an appropriate choice of the 
matrix~$A$ will give a locally satisfiable instance~$Ax=b$ for any
right-hand side~$b$.  Entirely analogous claims have been known and
proved in the context of the proof complexity of propositional
resolution; indeed, our proof builds on the methods for resolution
width \cite{BenSassonWigderson2001}, and their relationship to
existential pebble games from \cite{Atserias2004, AtseriasDalmau2008}.

In the proof, we need the notion of a graph~$G$ that is a
\emph{bipartite unique-neighbour expander graph with
  parameters~$(m,n,d,s,\beta)$} where~$m,n,d$ and~$s$ are integer
parameters with~$s < n$ and~$\beta$ is a positive real number.  What
this means is that~$G$ is a bipartite graph with parts~$U$ and~$V$
with~$m$ and~$n$ vertices respectively; each~$u \in U$ has exactly~$d$
neighbours in~$V$; and for every~$T \subseteq U$ with~$|T| \leq s$ we
have~$|\partial T| \geq \beta |T|$, where~$\partial T$ denotes the set
of vertices in~$V$ that are \emph{unique neighbours} of~$T$; i.e.,
they are neighbours of a single vertex in~$T$.

\begin{lemma} \label{lem:linear-local-bound} For every integer~$r > 0$
there is a real~$\gamma > 0$ such that for every sufficiently large
integer~$n$ there is a matrix~$A \in \{0,1\}^{m \times n}$,
where~$m=rn$, such that each row of~$A$ has exactly three ones and,
for every vector~$b \in \{0,1\}^m$, the~$\prob{3XOR}$ instance~$Ax=b$
is~$k$-locally satisfiable for~$k \leq \gamma n$.
\end{lemma}

\begin{proof}
Fix an integer~$r > 0$ and reals~$\alpha > 0$ and~$\beta > 0$, and
let~$n_0$ be sufficiently large that for every~$n \geq n_0$ there
exists a graph~$G$ that is a bipartite unique-neighbour expander graph
with parameters~$(rn,n,3,\alpha n, \beta)$.  For the existence of such
graphs with these parameters see
\cite[Chapter~4]{VadhanSurvey}. Let~$A \in \{0,1\}^{U \times V}$ be
the incidence matrix of~$G$, where~$U = [m]$ and~$V = [n]$ are the two
sides of~$G$, for~$m = cn$.  For each ~$b = (b_u : u \in U) \in
\{0,1\}^U$, the~$\prob{3XOR}$ instance ~$Ax=b$ has one variable~$x_v$
for each~$v \in V$, and one equation ~$e_u :
x_{v_1(u)}+x_{v_2(u)}+x_{v_3(u)} = b_u$ for each~$u \in U$,
where~$v_1(u),v_2(u),v_3(u)$ are the three neighbours of~$u$ in~$G$.
We claim that every choice of~$b \in \{0,1\}^U$ gives that~$Ax = b$
is~$k$-locally satisfiable for~$k \leq \gamma n$ with~$\gamma =
\alpha\beta/9$.

\begin{claim} \label{claim:minimally}
For every~$b \in \{0,1\}^U$, every set of at most~$\alpha n$
equations from~$Ax=b$ is satisfiable.
\end{claim}

\begin{proof}
For each~$T \subseteq U$, let~$e_T$ be the set of equations that are
indexed by vertices in~$T$, and let~$v_T$ be the set of variables that
appear in~$e_T$.  We prove, by induction on~$t \leq \alpha n$, that
if~$T \subseteq U$ and~$|T|=t$, then there exists an assignment that
sets all the variables in~$v_R$ and that satisfies all the equations
in~$e_T$.  For~$t=0$ the claim is obvious.  Assume now that~$1 \leq t
\leq \alpha n$ and let~$T$ be a subset of~$U$ of
cardinality~$t$. Then~$|\partial T| \geq \beta|T| > 0$. Let~$v_0$ be
some element in~$\partial T$ and let~$u_0 \in T$ be the unique
neighbour of~$v_0$ in~$T$.  The induction hypothesis applied to~$S = T
\setminus \{u_0\}$ gives an assignment~$g$ that sets all the variables
in~$v_S$ and satisfies all the equations in~$e_S$.  The assignment~$g$
may assign some of the variables of the equation~$e_{u_0}$, but not
all, since~$v_0$ is not a neighbour of any vertex in~$S$. Let~$f$ be
the unique extension of~$g$ that first sets all the variables in~$v_T
\setminus (v_S \cup \{x_{v_0}\})$ to~$0$, and then sets~$x_{v_0}$ to
the unique value that satisfies the equation~$e_{u_0}$. This
assignment sets all the variables in~$v_T$ and satisfies all the
equations in~$e_T$. The proof is complete.
\end{proof}

\begin{claim}
For every~$b \in \{0,1\}^U$ and~$k \leq \gamma n$, the instance~$I$
is~$k$-locally satisfiable.
\end{claim}

\begin{proof}
If~$I$ is satisfiable, then Duplicator certainly has a winning
strategy and there is nothing to prove. Assume then that~$I$ is
unsatisfiable and let~$I'$ be a minimally unsatisfiable subsystem; a
subset of the equations of~$I$ that is unsatisfiable and every proper
subset of it is satisfiable.  For each equation~$e_u : x_{v_1(u)} +
x_{v_2(u)} + x_{v_3(u)} = b_u$ of~$I$, let~$F_u$ be the four
clauses~$\{x_{v_1(u)}^{(a_1)},x_{v_2(u)}^{(a_2)},x_{v_3(u)}^{(a_3)}\}$
with~$a_1,a_2,a_3 \in \FF{2}$ with~$a_1+a_2+a_3=b_u$, where~$z^{(a)}$
stands for the negative literal~$\neg z$ if~$a = 0$ and the positive
literal~$z$ if~$a = 1$. Let~$F$ be the 3CNF formula that is the union
of all the~$F_u$ as~$u$ ranges over~$U$. Observe that~$F$ is an
unsatisfiable 3CNF. We intend to apply Theorem~5.9
from~\cite{BenSassonWigderson2001} to it.

Let~$\mathscr{A}$ be the collection of all Boolean functions~$f_u :
\{0,1\}^V \rightarrow \{0,1\}$ defined by
\begin{equation*}
f_u(x_v : v \in V) = x_{v_1(u)} + x_{v_2(u)} + x_{v_3(u)} + b_u \mod
2,
\end{equation*}
for~$u \in U$. Each function in~$\mathscr{A}$ is sensitive in the
sense of Definition~5.5 from~\cite{BenSassonWigderson2001}, and
compatible with~$F$ in the sense of Definition~5.3
from~\cite{BenSassonWigderson2001}. Moreover, if~$\mathscr{A}_0
\subseteq \mathscr{A}$ is the set of functions that corresponds to the
minimally unsatisfiable subsystem~$I'$ of~$I$, then its
cardinality~$m_0$ satisfies~$m_0 > \alpha n$ by
Claim~\ref{claim:minimally}. It follows that the
expansion~$e(\mathscr{A})$ in the sense of Definition~5.8
from~\cite{BenSassonWigderson2001} is at least~$\alpha\beta n/3$. By
Theorem~5.9 in \cite{BenSassonWigderson2001}, every resolution
refutation of~$F$ requires width at least~$e\alpha n/3$, and hence at
least~$3k$ since~$k \leq \gamma n = e\alpha n/9$. By Theorem~2 in
\cite{AtseriasDalmau2008}, Duplicator has a winning strategy for the
existential~$3k$-pebble game played on the structures~$F$ and the
constraint language~$\Gamma_{\prob{3SAT}}$ of~$\prob{3SAT}$, in the
second encoding discussed in Section~\ref{sec:encodings}. We use this
winning strategy to design a winning strategy for Duplicator in the
existential~$k$-pebble game played on~$I$ and~$\Gamma_{\prob{3XOR}}$.

While playing the game on~$I$, Duplicator plays the game on~$F$ on the
side and keeps the invariant that each pebbled variable in the game
on~$I$ is also pebbled in the side game, and each pebbled equation in
the game on~$I$ has its three variables pebbled in the side
game. Whenever a new variable is pebbled in the game on~$I$,
Duplicator pebbles the same variable in the side game, and copies the
answer from its strategy on it. Whenever a new equation is pebbled in
the game on~$I$, Duplicator pebbles its three variables in the side
game, and answers the pebbled equation accordingly from its
strategy. Since at each position of the game on~$I$ there are no more
than~$k$ pebbles on the board, at each time during the simulation the
side game has no more than~$3k$ pebbles on the board.  This shows that
the simulation can be carried on forever and the proof is complete.
\end{proof}

This completes the proof of Lemma~\ref{lem:linear-local-bound}.
\end{proof}

We can now prove our first two gap theorems.

\begin{theorem}\label{thm:3lin-onesided}
For every real~$\epsilon > 0$, if~$\class{C}$ is the collection
of~$\prob{3XOR}$ instances that are satisfiable and~$\class{D}$ is the
collection of~$\prob{3XOR}$ instances that are
not~$(1/2+\epsilon)$-satisfiable, then~$\class{C}$ and~$\class{D}$ are
not~$\Ck$-separable for any~$k = k(n)$ such that~$k(n) = o(n)$.
\end{theorem}
\begin{proof}
By combining Lemma~\ref{lem:linear-local-bound} with
Lemma~\ref{lem:both-extremal}, there is a family of systems ~$(S_k)_{k
  \geq 1}$ with~$O(k)$ variables and equations such that ~$G(S_k)$ is
not~$(1/2+\epsilon)$-satisfiable but~$S_k$ is ~$k$-locally
satisfiable. Let~$I_k^1 = G(S_k)$ and~$I_k^0 = G(S_k^0)$. Then~$I_k^0
\Cequiv{k} I_k^1$ by Lemma~\ref{lem:CFI}. Moreover, by the first part
of Lemma~\ref{lem:both-lift}, the instance~$I_k^0$ is satisfiable
while, by choice, the instance~$I_k^1$ is
not~$(1/2+\epsilon)$-satisfiable.  Since~$I_k^0$ and~$I_k^1$ have two
variables for each variable in~$S_k$ and eight equations for each
equation in~$S_k$, they also have~$O(k)$ variables and equations and
the result follows.
\end{proof}

\begin{theorem}\label{thm:3sat-onesided}
For every real~$\epsilon > 0$, if~$\class{C}$ is the collection
of~$\prob{3SAT}$ instances that are satisfiable and~$\class{D}$ is the
collection of~$\prob{3SAT}$ instances that are
not~$(7/8+\epsilon)$-satisfiable, then~$\class{C}$ and~$\class{D}$ are
not~$\Ck$-separable for any~$k = k(n)$ such that~$k(n) = o(n)$.
\end{theorem}
\begin{proof}
Consider the reduction~$\Theta$ from~$\prob{3XOR}$ to~$\prob{3SAT}$
that translates each equation into a conjunction of four clauses.
Thus~$x+y+z = d$ becomes four clauses~$\{x^{(a)},y^{(b)},z^{(c)}\}$
with~$a,b,c \in \FF{2}$ and~$a+b+c=d$, where~$z^{(e)}$ stands for the
negative literal~$\neg z$ if~$e = 0$ and the positive literal~$z$
if~$e = 1$.  This is easily defined in first-order logic.  As the set
of variables in~$I$ is the same as in~$\Theta(I)$, it is linearly
bounded. We claim that applying~$\Theta$ to
Theorem~\ref{thm:3lin-onesided} with~$\epsilon$ reset to~$\epsilon/4$
gives the theorem through Lemma~\ref{lem:reduction}.  First, it is
clear that if~$I$ is a~$\prob{3XOR}$ instance that is satisfiable,
then~$\Theta(I)$ is also satisfiable.  Now, suppose that~$I$ is a
system of~$m$ equations that is not~$(1/2+\epsilon/4)$-satisfiable,
and let~$g$ be an assignment of truth values to the variables~$X$
of~$\Theta(I)$.  Applied to~$I$, the assignment~$g$ falsifies at
least~$(1/2-\epsilon/4)m$ of the equations.  For each equation,~$g$
must falsify at least one of the four corresponding clauses
in~$\Theta(I)$.  Thus,~$g$ falsifies at least~$(1/2-\epsilon/4)m$
clauses in~$\Theta(I)$ and so satisfies at most~$4m -
(1/2-\epsilon/4)m = (7/8+\epsilon)\cdot 4m$ of the~$4m$ clauses.
\end{proof}

To formulate the consequences of these two theorems for~$\FPC$
definability, it is useful to introduce some terminology.  Say that a
\emph{term}~$\eta$ of~$\FPC$
\emph{$\delta$-approximates}~$\prob{MAX 3XOR}$, where~$0 < \delta \leq
1$, if whenever~$I$ is an instance of~$\prob{3XOR}$ in which a maximum
of~$m^*$ clauses are simultaneously satisfiable, then the
interpretation of~$\eta^I$ of~$\eta$ in~$I$ is a value such that~$m^*
\geq \eta^I \geq \delta m^*$.  The notion of a term
of~$\FPC$~$\delta$-approximating~$\prob{MAX 3SAT}$ is defined
similarly.

\begin{corollary}
For any~$\epsilon > 0$,
  \begin{enumerate} \itemsep=0pt
  \item there is no term of~$\FPC$ that~$(1/2 +
    \epsilon)$-approximates~$\prob{MAX 3XOR}$; and 
  \item  there is no term of~$\FPC$ that~$(7/8 +
    \epsilon)$-approximates~$\prob{MAX 3SAT}$.
  \end{enumerate}
\end{corollary}
\begin{proof}
If there were such a term in case (1), we would obtain an~$\FPC$
sentence defining a class~$\class{E}$ of counting width bounded by a
constant which separates the class of~$\prob{3XOR}$ instances that are
satisfiable from those that are not ~$(1/2 +\epsilon)$-satisfiable,
constradicting Theorem~\ref{thm:3lin-onesided}.  The analogous
situation holds in case (2) and Theorem~\ref{thm:3sat-onesided}.
\end{proof}

\section{Long Code Reductions} \label{sec:amplifying}

In this section we show that certain reductions from the theory of
inapproximability of~$\prob{MAX 3XOR}$ and~$\prob{MAX 3SAT}$ can be
expressed as~$\FO$-interpretations.  While no reduction can provide an
improvement on the already optimal inapproximability results that are
implied by Theorems~\ref{thm:3lin-onesided}
and~\ref{thm:3sat-onesided}, these~$\FO$-interpretations have the
merit of providing optimal gap pairs starting at \emph{any} initial
gap pair, provided the initial gap pair exhibits any constant gap
separation whatsoever. In addition, the details of
the~$\FO$-interpretations that we work out here will also be useful
when we discuss the reductions to the vertex-cover problem in the next
section.

\subsection{Parallel repetition}

We begin by defining the~$\prob{LABEL COVER}$ problem, a standard
problem in the study of hardness of approximation.  Indeed, it is described in the
textbook~\cite[p.~494]{AroraBarakBook} as being ``ubiquitous'' in the
PCP literature.  For a full discussion of the problem,
see~\cite[Chap.~22]{AroraBarakBook}, where it is called~$2\prob{CSP}_W$
with the projection property.

An instance~$I$ of the~$\prob{LABEL COVER}$ problem is given by two
disjoint sets of variables~$U$ and~$V$ with domains of values~$A$
and~$B$, respectively, a predicate~$P : U \times V \times A \times B
\rightarrow \{0,1\}$, and an assignment of weights~$W : U \times V
\rightarrow \nats$. If all the non-zero weights~$W(u,v)$ are equal,
then the instance is said to have \emph{uniform weights}. If for
all~$u \in U$ the sums~$W(u) := \sum_{v \in V} W(u,v)$ of incident
weights are equal, then the instance is called \emph{left-regular}. A
\emph{right-regular} instance is defined analogously in terms of~$W(v)
:= \sum_{u \in U} W(u,v)$.  The instance is a \emph{projection game}
if for every~$(u,v) \in U \times V$ with~$W(u,v) \not= 0$ it holds
that for every~$a \in A$ there is exactly one~$b \in B$
satisfying~$P(u,v,a,b) = 1$.  It is called a \emph{unique game}
if~$|A|=|B|$ and it is a projection game both ways: from~$A$ to~$B$,
and from~$B$ to~$A$. The instance is said to have
parameters~$(m,n,p,q)$ if~$|U|=m$,~$|V|=n$,~$|A|=p$ and~$|B|=q$. Its
\emph{domain size} is~$p+q$.

A value-assignment for an instance~$I$ is a pair of functions~$f : U
\rightarrow A$ and~$g : V \rightarrow B$. The weight~$v(f,g)$ of the
value-assignment~$(f,g)$ is the total weight of the pairs~$(u,v) \in U
\times V$ satisfying the constraint~$P(u,v,f(u),g(v)) = 1$; i.e.,
\begin{equation}
v(f,g) = \sum_{(u,v) \in U\times V} W(u,v) P(u,v,f(u),g(v)).
\end{equation}  
For~$c \in [0,1]$, we say that the instance is~$c$-satisfiable if
there is a value-assignment whose weight is at least~$c\cdot W_0$,
where~$W_0 = \sum_{(u,v) \in U \times V} W(u,v)$ is the maximum
possible weight. We call it satisfiable if it is~$1$-satisfiable.

The \emph{bipartite reduction} takes an instance~$I$ of~$\prob{3XOR}$
and produces a projection game instance~$L(I)$ of~$\prob{LABEL COVER}$
defined as follows. The sets~$U$ and~$V$ are the set of equations
in~$I$ and the set of variables in~$I$, respectively.  The
weight~$W(u,v)$ is~$1$ if~$v$ is one of the variables in the
equation~$u$, and~$0$ otherwise. The domains of values associated
to~$U$ and~$V$ are~$A = \{(a_1,a_2,a_3) \in \FF{2}^3 : a_1+a_2+a_3 =
0\}$ and~$B = \FF{2}$, respectively.  The predicate~$P$ associates to
the pair~$(u,v)$, where~$u$ is the equation~$v_1 + v_2 + v_3 = b$
and~$v = v_i$ for~$i \in \{1,2,3\}$, the set of
pairs~$((a_1,a_2,a_3),a) \in A \times B$ satisfying~$a = a_i + b$.  In
other words, ~$P(u,v,(a_1,a_2,a_3),a) = 1$ if, and only if,~$v$
appears in the equation~$u$, and if~$u$ is~$v_1+v_2+v_3=b$ and~$v =
v_i$, then the (partial) assignment~$\{ v_1 \mapsto a_1+b, v_2 \mapsto
a_2+b, v_3 \mapsto a_3+b \}$, which satisfies the
equation~$v_1+v_2+v_3 = b$ by construction, agrees with the (partial)
assignment~$\{ v_i \mapsto a \}$. Clearly, this defines a projection
game.

\begin{lemma} \label{lem:bipartitereduction}
For every instance~$I$ of~$\prob{3XOR}$ and every~$c,s \in [0,1]$, the
following hold:
\begin{enumerate} \itemsep=0pt
\item if~$I$ is~$c$-satisfiable, then~$L(I)$ is~$c$-satisfiable,
\item if~$I$ is not~$s$-satisfiable, then~$L(I)$ is not~$(s+2)/3$-satisfiable.
\end{enumerate}
Moreover,~$L(I)$ is a left-regular projection game that has uniform
weights.
\end{lemma}

\begin{proof}
Let~$m$ be the number of equations in~$I$, so~$L(I)$ has exactly~$3m$
pairs~$(u,v)$ of unit weight. Such pairs are called constraints. For
proving~1, let~$h$ be an assignment for~$I$ that satisfies at
least~$cm$ of the~$m$ equations in~$I$. For each equation~$u$ in~$I$,
say~$v_1+v_2+v_3=b$, define~$f(u) = (h(v_1)+b,h(v_2)+b,h(v_3)+b)$
if~$h$ satisfies~$v_1+v_2+v_3=b$, and define~$f(u) = (0,0,0)$
otherwise. For each variable~$v$ in~$I$, define~$g(v) = h(v)$. Each
equation in~$I$ gives rise to exactly three constraints in~$L(I)$, and
if the equation is satisfied by ~$h$, then all three constraints
associated to it in~$L(I)$ are satisfied by~$(f,g)$. Thus~$(f,g)$
satisfies at least~$3cm$ of the~$3m$ constraints in~$L(I)$, so~$L(I)$
is ~$c$-satisfiable. For proving~2, let~$(f,g)$ be an assigment
for~$L(I)$ that satisfies at least~$(s+2)m$ of the~$3m$ constraints
in~$L(I)$. For each variable~$v$ in~$I$, define~$h(v) = g(v)$. Let~$t$
be the number of equations of~$I$ that are satisfied by~$h$. In terms
of~$t$, the assignment~$(f,g)$ satisfies at most~$3t+2(m-t)$ of
the~$3m$ constraints of~$L(I)$. Thus~$t \geq sm$, so~$I$
is~$s$-satisfiable.
\end{proof}

The \emph{parallel repetition reduction} takes an instance~$I$
of~$\prob{LABEL COVER}$, and an integer~$t \geq 1$, and produces
another instance~$R(I,t)$ of~$\prob{LABEL COVER}$ defined as
follows. Let~$U$ and~$V$ be the sets of variables in~$I$ and let~$W :
U \times V \rightarrow \nats$ be the weight assignment. The sets of
variables of~$R(I,t)$ are~$U^t$ and~$V^t$. For~$\vect{u} =
(u_1,\ldots,u_t) \in U^t$ and~$\vect{v} = (v_1,\ldots,v_t) \in V^t$,
the weight~$W(\vect{u},\vect{v})$ is defined as~$\prod_{i=1}^t
W(u_i,v_i)$. If~$A$ and~$B$ are the domains of values associated
to~$U$ and~$V$, then the domains of values associated to~$U^t$
and~$V^t$ are~$A^t$ and~$B^t$ respectively.  For~$\vect{u} =
(u_1,\ldots,u_t) \in U^t$,~$\vect{v} = (v_1,\ldots,v_t) \in
V^t$,~$\vect{a} = (a_1,\ldots,a_t) \in A^t$ and~$\vect{b} =
(b_1,\ldots,b_t) \in B^t$, the
predicate~$P(\vect{u},\vect{v},\vect{a},\vect{b})$ is defined
as~$\prod_{i=1}^t P(u_i,v_i,a_i,b_i)$. Observe that this definition
guarantees that if~$I$ is a projection game, then so is~$R(I,t)$.

\begin{theorem}[Parallel Repetition Theorem \cite{Raz,Holenstein2007}]
  \label{thm:parallelrepetition}
There exists a constant~$\alpha > 0$ such that for every instance~$I$
of~$\prob{LABEL COVER}$ with domain size at most~$d \geq 1$, every~$s
\in [0,1]$ and every~$t \geq 1$ the following hold:
  \begin{enumerate}\itemsep=0pt 
  \item if~$I$ is satisfiable, then~$R(I,t)$ is satisfiable,
  \item if~$I$ is not~$s$-satisfiable, then~$R(I,t)$ is
    not~$(1-(1-s)^3)^{\alpha t/d}$-satisfiable.
  \end{enumerate}
Moreover, if~$I$ is a projection game, left-regular, right-regular, or
has uniform weights, then so is~$R(I,t)$.
\end{theorem}

Although it is the case that the bipartite and the parallel repetition
reductions are both~$\FO$-interpretations, we do not need to formulate
this.  Instead, we show the~$\FO$-definability of the composition of
these reductions with the long-code reductions that we discuss next.

\subsection{First long-code reduction}

The \emph{first long-code reduction} takes a
projection game instance~$I$ of~$\prob{LABEL COVER}$ and a
rational~$\epsilon \in [0,1]$ and produces an instance~$C(I,\epsilon)$
of~$\prob{3XOR}$ defined as follows. Let~$U$ and~$V$ be the sets of
variables of sizes~$m$ and~$n$, respectively, with associated domains
of values~$A = [p]$ and~$B = [q]$, let~$W : U \times V \rightarrow
\nats$ be the weight assignment, let~$P : U\times V\times A\times B
\rightarrow \{0,1\}$ be the predicate of~$I$, and for each~$(u,v) \in
U \times V$ with~$W(u,v) \not= 0$ and each~$a \in A$
let~$\pi_{u,v}(a)$ be the unique value~$b \in B$ that
satisfies~$P(u,v,a,b) = 1$. The existence of such a
function~$\pi_{u,v} : A \rightarrow B$ is guaranteed from the
assumption that~$I$ is a projection game. The set of variables
of~$C(I,\epsilon)$ includes one variable~$u(a)$ for each~$u \in U$
and~$a \in \FF{2}^{p-1}$, and one variable~$v(b)$ for each~$v \in V$
and~$b \in \FF{2}^{q-1}$, for a total of~$m 2^{p-1} + n 2^{q-1}$
variables. Before we are able to define the set of equations
of~$C(I,\epsilon)$ we need a piece of notation. For a vector~$z =
(z_1,\ldots,z_d) \in \FF{2}^d$ of dimension~$d \geq 2$, we write~$S(z)
= z_d$ and~$F(z) = (z_1 + S(z),\ldots,z_{d-1} + S(z))$. Note
that~$S(z)$ is a single field element, and~$F(z)$ is a vector of
dimension~$d-1$.  With this notation, the set of equations
of~$C(I,\epsilon)$ includes~$W(u,v) \cdot M^q \cdot \epsilon^D \cdot
(1-\epsilon)^{q-D}$ copies of the equation~$v(F(x)) + u(F(y)) +
u(F(z)) = S(x) + S(y) + S(z)$ for each~$(u,v) \in U \times V$, each~$x
\in \FF{2}^{q}$ and each~$y,z \in \FF{2}^{p}$, where~$M$ is the
denominator of~$\epsilon = N/M$ reduced to lowest terms,~$D$ is the
number of positions~$i \in [p]$ such that~$z_i \not= x_{\pi(i)} +
y_i$, and~$\pi = \pi_{u,v}$ if~$W(u,v) \not= 0$.

\begin{theorem}[H{\aa}stad 3-Query Linear Test \cite{Hastad2001}]
\label{thm:hastad}
For every~$s,\epsilon \in [0,1]$ with~$\epsilon > 0$ and~$s > 0$ and
every projection game instance~$I$ of~$\prob{LABEL COVER}$, the
following hold:
\begin{enumerate}\itemsep=0pt 
\item if~$I$ is satisfiable, then~$C(I,\epsilon)$
  is~$(1-\epsilon)$-satisfiable,
\item if~$I$ is not~$s$-satisfiable, then~$C(I,\epsilon)$ is
not~$(1/2+(s/\epsilon)^{1/2}/4)$-satisfiable.
\end{enumerate}
\end{theorem}

\noindent 
The proof of Theorem~\ref{thm:hastad} follows from Lemmas 5.1 and 5.2
in \cite{Hastad2001}. In order to see this, we need to explain how our
notation matches the one in \cite{Hastad2001}.  Besides the obvious
and minor correspondance between multiplicative and additive notation
for~$\FF{2}$, with~$-1 \leftrightarrow 1$ and~$+1 \leftrightarrow 0$,
there are three other noticeable differences between the statement of
Theorem~\ref{thm:hastad} and the statements of Lemmas~5.1 and~5.2
in~\cite{Hastad2001}.

The first difference is that Theorem~\ref{thm:hastad} applies to
arbitrary projection game instances of~$\prob{LABEL COVER}$, while the
statements in~\cite{Hastad2001} are phrased only for the special cases
of the problem that result from applying the parallel repetition
construction to a suitable bipartite reduction applied to
a~$\prob{3SAT}$ instance. We chose to formulate
Theorem~\ref{thm:hastad} in this more general and modular form because
this is what the proofs of Lemmas~5.1 and~5.2 in~\cite{Hastad2001}
show, and also because this is how more recent expositions of these
results are presented (see, e.g., \cite{AroraBarakBook}).

The second difference is that the conclusion of our statement is
phrased in terms of the~$c$-satisfiability of a~$\prob{3XOR}$
instance, while the statements of Lemmas~5.1 and~5.2 in
\cite{Hastad2001} are phrased in terms of the acceptance rate of a
probabilistic test that has the following form: given access to
certain \emph{tables}~$A_u$ and~$A_v$, with~$\FF{2}$
entries~$\{A_u(x)\}_{x \in I}$ and~$\{ A_v(y) \}_{ y \in J }$ for
certain index sets~$I$ and~$J$, respectively, choose a random
3-variables parity test on the~$A_u(x)$ and~$A_v(y)$ entries under a
well-designed special-purpose distribution, and check if it is
satisfied. This difference is only notational and minor: our instance
of~$\prob{XOR}$ is built by viewing the~$A_u(x)$ and~$A_v(y)$ entries
as variables~$u(x)$ and~$v(y)$, and assigning weight to each
3-variable parity equation on these variables proportionally to the
probability that it is checked by the probabilistic test on the~$A_u$
and~$A_v$ tables. With this change,~$c$-satisfiability of the instance
translates into the probability of acceptance of the test being at
least~$c$, and vice-versa.

The third difference in the notation is that our variables~$u(x)$
and~$v(y)$, and the corresponding entries~$A_u(x)$ and~$A_v(y)$ of the
tables~$A_u$ and~$A_v$, are indexed by~$\mathbb{F}_2^{p-1}$
and~$\mathbb{F}_2^{q-1}$ instead of the more
natural~$\mathbb{F}_2^{p}$ and~$\mathbb{F}_2^q$, respectively. This is
due to the fact that we implement the operations of \emph{folding over
  true} and \emph{conditioning upon~$h$} from~\cite{Hastad2001}
directly in our construction. In other words, our tables~$A_u$
and~$A_v$ are what~\cite{Hastad2001}
calls~$A_{W,h,\text{\emph{true}}}$ and~$A_{U,\text{\emph{true}}}$,
respectively. Folding over true as in~$A_{U,\text{\emph{true}}}$ is
achieved for~$A_v$ through the notation~$S(z)$ and~$F(z)$ defined
above: we chose to partition~$\FF{2}^p$ into~$2^{p-1}$ pairs of the
form~$(z,0),(F((z,1)),1)$, as~$z$ ranges over~$\FF{2}^{p-1}$, and view
an arbitrary~$A_v : \FF{2}^{p-1} \rightarrow \FF{2}$ as representing
the function~$A'_v : \FF{2}^p \rightarrow \FF{2}$ defined by~$A'_v(z)
= A_v(F(z)) + S(z)$ for every~$z \in \FF{2}^p$. It is straightfoward
to see that~$A'_v$ is folded over true, in the definition
of~\cite{Hastad2001}, by construction.

Conditioning upon~$h$ as in~$A_{W,h,\text{\emph{true}}}$ for~$A_u$ is
achieved through the same mechanism as folding over true with the
additional observation that the operation of conditioning upon~$h$ is
necessary only if the instance of~$\prob{LABEL COVER}$ fails to
satisfy the property that for every~$(u,v) \in U \times V$ and
every~$a \in A$ there is at least one~$b \in B$ that satisfies the
predicate~$P(u,v,a,b)$. When this is the case, one defines~$h =
h_{u,v} : A \rightarrow \{0,1\}$ as the predicate indicating if a
given~$a$ has at least one~$b$ that satisfies~$P(u,v,a,b)$, and
\emph{conditions the table~$A_u$ upon~$h$}. In our case we do not
require this since the given instance of~$\prob{LABEL COVER}$ is a
projection game instance, and, in particular, for every~$a$ there is
exactly one~$b$, and hence at least one~$b$, such that~$P(u,v,a,b) =
1$; i.e.,~$h = h_{u,v}$ is the constant~$1$ predicate. It should be
added that the reason why we can assume that~$I$ is a projection game
instance is that our bipartite reduction is designed in such a way
that the values~$a$ in~$A$ are partial assignments that always satisfy
the corresponding constraints~$u$ in~$U$. In constrast,
in~\cite{Hastad2001} the values are taken as arbitrary truth
assignments to the variables of a collection of clauses, and not all
such assignments satisfy all the clauses. Our exposition is again more
modular and also matches more recent expositions of the results
in~\cite{Hastad2001} (again, see, e.g.,~\cite{AroraBarakBook}).

With this notational correspondence, it is now easy to see that
Lemma~5.1 in~\cite{Hastad2001} gives the first claim in
Theorem~\ref{thm:hastad}, and Lemma~5.2 in~\cite{Hastad2001} applied
with~$\delta = (s/\epsilon)^{1/2}/4$ gives the second claim in
Theorem~\ref{thm:hastad}.

Next, by composing Lemma~\ref{lem:bipartitereduction},
Theorem~\ref{thm:parallelrepetition}, and Theorem~\ref{thm:hastad}
with the appropriate parameters we get the following:

\begin{theorem} \label{thm:gapamplireduction}
For every~$s,\epsilon \in [0,1]$ with~$0 < s < 1$ and~$\epsilon > 0$,
there is an FO-interpretation~$\Theta$ that maps instances
of~$\prob{3XOR}$ to instances of~$\prob{3XOR}$ in such a way that, for
every~$\prob{3XOR}$ instance~$I$ the following hold:
\begin{enumerate}\itemsep=0pt 
\item if~$I$ is satisfiable, then~$\Theta(I)$
  is~$(1-\epsilon)$-satisfiable,
\item if~$I$ is not~$s$-satisfiable, then~$\Theta(I)$ is
  not~$(1/2+\epsilon)$-satisfiable.
\end{enumerate}
\end{theorem}

\begin{proof}
First we define~$\Theta(I)$ and then check that this definition is an
FO-interpretation. In anticipation for the proof, let~$t$ be a large
enough integer so that the following inequality holds:
\begin{equation}
(1-(1-(s+2)/3)^3)^{\alpha t/6} \leq 16\epsilon^3,
\end{equation} 
where~$\alpha$ is the constant in
Theorem~\ref{thm:parallelrepetition}. Such a~$t$ exists because~$s <
1$ and~$\epsilon > 0$. Apply the bipartite reduction to~$I$ to obtain
the instance~$I' = L(I)$ from
Lemma~\ref{lem:bipartitereduction}. Observe that the domain size~$d$
of~$I'$ is~$|A|+|B|=6$. Next apply the parallel repetition reduction
to~$I'$ with parameter~$t$ to obtain a new instance~$I''$. Finally
apply the long-code reduction to~$I''$ with parameter~$\epsilon$ to
obtain the system~$I'''$. The parameters were chosen in a way that the
system~$I'''$ satisfies properties~1 and~2, through
Theorem~\ref{thm:hastad}.

It remains to argue that~$I'''$ can be produced from~$I$ by an
FO-interpretation.  To define~$I'$ from~$I$ there is no difficulty at
all: the FO-interpretation is even linear.  To define~$I''$ from~$I'$
we note that~$t$ is a constant, and that the weights~$W(u,v)$ of~$I'$
are~$0$ or~$1$, so again there is no difficulty. In this case the
FO-interpretation has dimension~$t$, and it is~$n^t$-bounded.  To
define~$I'''$ from~$I''$ we note that the domain sizes~$p$ and~$q$ of
the instance~$I''$ are constants, indeed~$p = 4^t$ and~$q = 2^t$. This
means that there are~$|U| \cdot 2^{p-1}$ variables of type~$u(a)$,
and~$|V| \cdot 2^{q-1}$ variables of type~$v(b)$, and these are
constant multiples of~$|U|$ and~$|V|$, respectively. Such domains
are~$\FO$-definable by the method of finite expansions (see
Section~\ref{sec:preliminaries}). Finally, since the weights~$W(u,v)$
of~$I''$ are still zeros or ones and both~$\epsilon$ and~$q$ are
constants, the multiplicities of the equations of~$I'''$ are also
constants, and hence~$\FO$-definable.
\end{proof}

It is useful to compare Theorem~\ref{thm:gapamplireduction} with
Lemma~\ref{lem:both-lift}. Both statements are reductions that
take~$\prob{3XOR}$ instances to~$\prob{3XOR}$ instances, and they both
preserve gaps.  But the reductions differ in what happens to
satisfiable instances.  For statement~\emph{1}, in which the extreme
case is~$c = 1$, the reduction in Lemma~\ref{lem:both-lift} preserves
this extremality exactly.  In contrast, the reduction in
Theorem~\ref{thm:gapamplireduction} incurs a vanishing~$\epsilon$ loss
as it produces instances that are only~$(1-\epsilon)$-satisfiable.

\subsection{Second long-code reduction}

The \emph{second long-code reduction} takes a projection game
instance~$I$ of~$\prob{LABEL COVER}$ and a rational~$\delta \in [0,1]$
and produces an instance~$D(I,\delta)$ of~$\prob{3SAT}$ defined as
follows. Before we define~$D(I,\delta)$, let us define an intermediate
instance~$D'(I,\epsilon)$ of~$\prob{3SAT}$ that takes a different
parameter~$\epsilon \in [0,1]$.
Let~$U$,~$V$,~$m$,~$n$,~$A$,~$B$,~$p$,~$q$,~$W$,~$P$,
and~$\pi_{u,v}(a)$ be as in the first long-code reduction. The set of
variables of~$D(I,\epsilon)$ is defined as in the first long-code
reduction: a variable~$u(a)$ for each~$u \in U$ and each~$a \in
\FF{2}^{p-1}$, and a variable~$v(b)$ for each~$v \in V$ and each~$b
\in \FF{2}^{q-1}$. We also use the \emph{folding} notation~$F(z)$
and~$S(z)$ from the first long-code reduction.  Now the
instance~$D'(I,\epsilon)$ includes~$W(u,v)\cdot M^q \cdot \epsilon^D
\cdot (1-\epsilon)^{E-D} \cdot H$ copies of the clause~$\{
v(F(x))^{(S(x))} , u(F(y))^{(S(y))} , u(F(z))^{(S(z))} \}$ for
each~$(u,v) \in U \times V$, each~$x \in \FF{2}^q$ and each~$y,z \in
\FF{2}^p$, where~$M$ is the denominator of~$\epsilon = N/M$ reduced to
lowest terms,~$E$ is the number of positions~$i \in [p]$
with~$x_{\pi(i)} = 1$ and~$D$ is the number of positions~$i \in [p]$
with~$x_{\pi(i)} = 1$ and~$z_i \not= y_i$ for~$\pi = \pi_{u,v}$
if~$W(u,v) \not= 0$, while~$H \in \{0,1\}$ is the indicator for the
event that in each position~$i \in [p]$ with~$x_{\pi(i)} = 0$ we
have~$z_i \not= y_i$.  Finally, to define the instance~$D(I,\delta)$,
set~$t = \lceil{\delta^{-1}}\rceil$ and~$\epsilon_1 = \delta$,
and~$\epsilon_{i+1} = \delta^{71} 2^{-35} \epsilon_{i}$ for~$i =
1,\ldots,t-1$, and let the instance be~$\bigcup_{i=1}^t
D'(I,\epsilon_i)$.

\begin{theorem}[H{\aa}stad 3-Query Disjunction Test \cite{Hastad2001}]
\label{thm:hastad2}
There exists~$s_0 > 0$ such that for every~$s \in [0,1]$ with~$0 < s <
s_0$ and every projection game instance~$I$ of~$\prob{LABEL COVER}$
the following hold:
\begin{enumerate}\itemsep=0pt 
\item if~$I$ is satisfiable, then~$C(I,\epsilon)$ is
satisfiable,
\item if~$I$ is not~$s$-satisfiable, then~$C(I,\epsilon)$ is not~$(7/8
  + \log_2(1/s)^{-1/2})$-satisfiable.
\end{enumerate}
\end{theorem}

\noindent 
For the proof of Theorem~\ref{thm:hastad2}, see Lemmas 6.12 and 6.13
in \cite{Hastad2001}. As in the first long-code reduction, some
explanation is needed for seeing this.

Besides the notational differences that were already pointed out in
the first long-code reduction, the second long-code reduction adds the
following. First, the constants~$71$ and~$35$ in the definition
of~$\epsilon_{i+1}$ come from setting~$c = 1/35$ in the definition of
Test F3S$^\delta(u)$ in~\cite{Hastad2001}. According to Lemma~6.9
in~\cite{Hastad2001}, this is an acceptable setting of~$c$. Second,
the constant~$s_0 > 0$ in Theorem~\ref{thm:hastad2} is meant to be
chosen small enough so as to ensure that, for each~$s$ satisfying~$s <
s_0$, we have~$2^{-64\delta^{-2}/25} < 2^{-d
  \delta^{-1}\log_2(\delta^{-1})}$ for~$\delta =
8\log_2(1/s)^{-1/2}/5$, where~$d$ is the constant hidden in the
asymptotic~$O$-notation of Lemma~6.13 in~\cite{Hastad2001}.  Such
an~$s_0$ exists because~$N \log_2(N) = o(N^2)$ as~$N \rightarrow
+\infty$. With this notation, Lemma~6.12 in~\cite{Hastad2001} gives
point~1, and Lemma~6.13 in~\cite{Hastad2001} with~$\delta =
8\log_2(1/s)^{-1/2}/5$ gives point~2 in Theorem~\ref{thm:hastad2}.

By composing Lemma~\ref{lem:bipartitereduction},
Theorem~\ref{thm:parallelrepetition}, and Theorem~\ref{thm:hastad2}
with the appropriate parameters we get the following:

\begin{theorem} \label{thm:gapamplireduction2}
For every~$s,\epsilon \in [0,1]$ with~$0 < s < 1$ and~$\epsilon > 0$,
there is an FO-interpretation~$\Theta$ that maps instances
of~$\prob{3XOR}$ to instances of~$\prob{3SAT}$ in such a way that, for
every~$\prob{3XOR}$ instance~$I$ the following hold:
\begin{enumerate}\itemsep=0pt 
\item if~$I$ is satisfiable, then~$\Theta(I)$ is satisfiable,
\item if~$I$ is not~$s$-satisfiable, then~$\Theta(I)$ is
  not~$(7/8+\epsilon)$-satisfiable.
\end{enumerate}
\end{theorem}

\begin{proof}
First we define~$\Theta(I)$ and then check that this definition is an
FO-interpretation.  Let~$t$ be a large enough integer so that the
following inequality holds:
\begin{equation}
(1-(1-(s+2)/3)^3)^{\alpha t/6} \leq \min\{ 2^{-1/\epsilon^2}, s_0 \}
\end{equation} 
where~$\alpha$ is the constant in Theorem~\ref{thm:parallelrepetition}
and~$s_0 > 0$ is small enough as in Theorem~\ref{thm:hastad2}. Such
a~$t$ exists because~$s < 1$ and~$\epsilon > 0$ as well as~$s_0 >
0$. Apply the bipartite reduction to~$I$ to obtain the instance~$I' =
L(I)$ from Lemma~\ref{lem:bipartitereduction}. Observe that the domain
size~$d$ of~$I'$ is~$|A|+|B|=6$. Next apply the parallel repetition
reduction to~$I'$ with parameter~$t$ to obtain a new
instance~$I''$. Finally apply the second long-code reduction to~$I''$
to obtain the system~$I'''$. The parameters were chosen so that the
system~$I'''$ satisfies properties~1 and~2, through
Theorem~\ref{thm:hastad2}. As in the proof of
Theorem~\ref{thm:gapamplireduction} this reduction is FO-definable.
\end{proof}

This gives us another route to Theorem~\ref{thm:3sat-onesided}.

\section{Vertex Cover}

We investigate gap inexpressibility results for the vertex cover
problem~$\prob{VC}$ on graphs. Recall that a set~$X \subseteq V$ of
vertices in a graph~$G=(V,E)$ is a \emph{vertex cover} if every edge
in~$E$ has at least one of its endpoints in~$X$.  If the graph comes
with a weight function~$w : V \rightarrow \mathbb{R}^+$, then the
weight of~$X$ is the sum of the weights of the vertices in~$X$. If the
weights of the vertices are omitted in the specification of the graph,
then all the vertices are assumed to have unit weight. The problem of
finding the minimum weight vertex cover in a graph is a
classic~$\NP$-complete problem.

In the following we write~$\vc{G}$ for the weight of a minimum weight
vertex cover, and~$\vcd{G} := \vc{G}/W_0$, where~$W_0 := \sum_{v \in
  V} w(v)$, for the \emph{vertex cover density}. Analogously, we
write~$\is{G}$ for the weight of a maximum weight independent set,
and~$\isd{G} := \is{G}/W_0$. Clearly~$\vcd{G} = 1 - \isd{G}$ holds for
all weighted graphs.

\subsection{Direct reductions}

The standard reduction that proves the~$\NP$-completeness of the
vertex cover problem (see, e.g.~\cite[Thm.~9.4]{Pap94}) takes an
instance~$I$ of~$\prob{3SAT}$ with~$n$ variables and~$m$ clauses and
gives a graph~$G$ with~$3m$ vertices in which the minimum vertex cover
has size exactly~$2 c m$, if~$c m$ is the maximum number of clauses
in~$I$ that can be simultaneously satisfied.  It is also easy to see
that this reduction can be given as an ~$\FO$-interpretation.  This
interpretation is linearly bounded and therefore it follows from
Theorem~\ref{thm:3sat-onesided} and Lemma~\ref{lem:reduction} that for
any~$\epsilon > 0$ the collection of graphs~$G$ with~$\vcd{G} \leq
7/12+\epsilon$ and the collection of graphs~$G$ with~$\vcd{G} \geq
2/3$ cannot be separated in~$\Ck$ for any~$k = o(n)$.  This has the
consequence that no approximation algorithm for the vertex cover
problem expressible in~$\FPC$ can achieve an approximation ratio
better than~$8/7$.

We can improve on this by considering instead the so-called FGLSS
reduction (see \cite{FGLSS1996}, and \cite{Goldreich2011} for what by
now became standard terminology) from~$\prob{3XOR}$ to vertex-cover,
which we describe next.

\begin{theorem}\label{thm:fglss}
There is a linearly-bounded first-order reduction~$G$ that takes an
instance~$I$ of~$\prob{3XOR}$ with~$m$ equations to a graph~$G(I)$
with~$4m$ vertices so that if~$m^*$ is the maximum number of equations
of~$I$ that can be simultaneously satisfied, then~$\vc{G} = 4m-m^*$.
\end{theorem}
\begin{proof}
For each equation~$x+y+z = b$ in~$I$, the graph~$G(I)$ has a 4-clique
of vertices, each labelled with a distinct assignment of values to the
three variables that make the equation true.  In addition, we have an
edge between any pair of vertices that are labelled by inconsistent
assignments.  It is easily seen that the largest independent set
in~$G(I)$ is obtained by taking an assignment~$g$ of values to the
variables of~$I$ that satisfies ~$m^*$ equations and, for each
satisfied equation, selecting the vertex in its 4-clique that is the
projection of~$g$.  This yields an independent set of size
exactly~$m^*$ and the result follows.
\end{proof}

From this, and Theorem~\ref{thm:3lin-onesided}, we immediately get the
following result.
\begin{corollary}\label{cor:Ck-VC1}
For any~$\epsilon > 0$, if~$\class{C}$ is the collection of graphs~$G$
with~$\vcd{G} \leq 3/4$ and~$\class{D}$ is the collection of
graphs~$G$ with~$\vcd{G} \geq 7/8-\epsilon$ then~$\class{C}$
and~$\class{D}$ are not~$\Ck$-separable for any~$k = k(n)$ such
that~$k(n) = o(n)$.
\end{corollary}

%


This improves the FPC inapproximability ratio from~$8/7$ to~$7/6$.
Better lower bounds on the approximation ratio are known under the
assumption that~$\PT\neq \NP$.  One such lower bound was achieved by
Dinur and Safra~\cite{DinurSafra} who showed that, under this
assumption, no polynomial-time algorithm for approximating vertex
cover can achieve an approximation ratio better than~$1.36$.  In the
next section we argue that this reduction is also an
FO-interpretation, so we get the same inapproximability ratio for
algorithms that are expressible in~$\FPC$, giving a strengthening of
Corollary~\ref{cor:Ck-VC1}.

\subsection{Dinur-Safra reduction}

As in the long-code reductions from Section~\ref{sec:amplifying}, this
reduction is also composed of several steps: we start with the
bipartite reduction, continue with the parallel repetition reduction,
then we apply an \emph{intermediate reduction} to a technical variant
of the independent set problem, and end with a long-code reduction
that is specially tailored for the vertex cover problem.

The \emph{intermediate reduction} takes a projection game instance~$I$
of~$\prob{LABEL COVER}$ as input and produces an undirected
graph~$G(I)$ defined as follows.  Let~$U$,~$V$,~$A$,~$B$,~$W$
and~$\pi_{u,v}$ determine the projection game instance~$I$. The set of
vertices of the graph~$G(I)$ is~$U \times A$.  There is an edge
between~$(u_1,a_1) \in U \times A$ and~$(u_2,a_2) \in U \times A$
in~$G(I)$ if, and only, if either~$u_1 = u_2$ and~$a_1 \not= a_2$,
or~$u_1 \not= u_2$ and there exists~$v \in V$ such that~$W(u_1,v) > 0$
and~$W(u_2,v) > 0$ and~$\pi_{u_1,v}(a_1) \not= \pi_{u_2,v}(a_2)$. This
defines~$G(I)$.  In the terminology of~\cite{DinurSafra}, the
graph~$G(I)$ is ~$(m,p)$-co-partite: its edge-set is the complement of
an~$m$-partite graph with all its parts of size~$p$.

For an undirected (unweighted) graph~$G$, recall that~$\is{G}$ denotes
the size of a largest independent set in~$G$. For an integer~$h \geq
2$, let~$\his{h}{G}$ denote the size of a largest subset of vertices
of~$G$ that does not contain any~$h$-clique. Note that~$\his{h}{G}
\geq \is{G}$, and~$\his{2}{G} = \is{G}$.

\begin{lemma} \label{lem:intermediatereduction} 
For every integer~$h \geq 2$, every~$s \in [0,1]$ with~$s > 0$ and
every instance~$I$ of ~$\prob{LABEL COVER}$ that is a left-regular
projection game with uniform weights and parameters~$(m,n,p,q)$, the
following hold:
\begin{enumerate} \itemsep=0pt
\item if~$I$ is satisfiable, then~$\is{G(I)} = m$,
\item if~$I$ is not~$s$-satisfiable, then~$\his{h}{G(I)} \leq sh^3 m$.
\end{enumerate}
Moreover,~$G(I)$ is an~$(m,p)$-co-partite graph.
\end{lemma}

\noindent Note that, in the statement of
Lemma~\ref{lem:intermediatereduction}, the completeness case (point
\emph{1}) is about~$\is{G}$ but the soundness case (point \emph{2}) is
about~$\his{h}{G}$. For the proof of
Lemma~\ref{lem:intermediatereduction}, see the proof of Theorem~2.1
in~\cite{DinurSafra}.

The \emph{vertex-cover long-code reduction} inputs a graph~$G =
(V,E)$, two rational parameters~$\epsilon > 0$ and~$p$ satisfying~$0 <
p < p_{\mathrm{max}} := (3-\sqrt{5})/2$, and two integer parameters~$h
\geq 2$ and~$r \geq 2$, and, if~$G$ is an~$(m,r)$-co-partite graph for
some~$m$, then it produces a (weighted) graph~$H(G,\epsilon,p,h; r)$
defined as follows. Let the vertex-set~$V$ of~$G$ be~$[m] \times [r]$,
so that~$\{(u,a) : a \in [r]\}$ forms an~$r$-clique for each~$u \in
[m]$. In abstract, the set of vertices of~$H(G,\epsilon,p,h; r)$ is
\begin{equation} 
W := \{ (B,\powerset{\kpowerset{\geq l_{1}}{B}}) : B \in
\kpowerset{=l}{V} \}, \label{eqn:defW}
\end{equation} 
where~$l = 2l_{1} \cdot r$, and~$l_{1}$ is an integer that depends
only on~$h$,~$\epsilon$, and~$p$, and is independent of~$r$, that is
set as in Definition~2.3 of~\cite{DinurSafra}. Here, and in the
following,~$\kpowerset{=k}{X}$ and~$\kpowerset{\geq k}{X}$ denote the
collections of subsets of~$X$ of size exactly~$k$ and size at
least~$k$, respectively, and~$\powerset{X}$ denotes the collection of
all subsets of~$X$. Thus, if~$n = mp$ is the number of vertices
of~$G$, then~$H(G,\epsilon,p,h;r)$ has~$\binom{n}{l} \cdot
2^{\sum_{i=l_{1}}^l \binom{l}{i}}$ vertices.  Since we want to be able
to show that for fixed~$r$,~$h$,~$\epsilon$ and~$p$ the
graph~$H(G,\epsilon,p,h;r)$ can be produced from~$G$ by
an~$\FO$-interpretation, we give an alternative presentation of the
set of vertices~$W$.

Let~$V^{l,\not=}$ denote the set of~$l$-tuples of pairwise distinct
elements from~$V$. Formally,
\begin{equation}
V^{l,\not=} := 
\{(u_1,\ldots,u_l) \in V^l : u_i \not= u_j \text{ for } i,j \in [l]
\text{ with } i \not= j \}. \label{eqn:notfinalaltdefW}
\end{equation} 
For each~$\vect{u} = (u_1,\ldots,u_l) \in V^{l,\not=}$,
let~$\sigma_{\vect{u}} : \{1,\ldots,l\} \rightarrow
\{u_1,\ldots,u_l\}$ be the natural bijection defined
by~$\sigma_{\vect{u}}(i) = u_i$ for~$i = 1,\ldots,l$. The set
\begin{equation}
W' := V^{l,\not=} \times \powerset{\kpowerset{\geq
    l_1}{[l]}} \label{eqn:proxy}
\end{equation}
is a good proxy for the set~$W$ through the identification
of~$\{1,\ldots,l\}$ and~$\{u_1,\ldots,u_l\}$ given
by~$\sigma_{\vect{u}}$.  Now, turning~$W'$ into a faithful copy of~$W$
is only a matter of taking a quotient with the appropriate equivalence
relation, as we do next.  

Let~$\sim$ be the equivalence relation
on~$V^l$ defined by~$(u_1,\ldots,u_l) \sim (v_1,\ldots,v_l)$ if and
only if for each~$i \in [l]$ there exists~$j \in [l]$ with~$v_j = u_i$
and for each~$j \in [l]$ there exists~$i \in [l]$ with~$u_i = v_j$.
Restricted to~$V^{l,\not=} \subseteq V^l$, this is still an
equivalence relation.  Moreover, whenever~$\vect{u} =
(u_1,\ldots,u_l)$ and ~$\vect{v} = (v_1,\ldots,v_l)$
are~$\sim$-equivalent tuples in ~$V^{l,\not=}$, there is a unique
permutation~$\pi \in S_l$ that sends~$\vect{u}$ to~$\vect{v}$; i.e.,
that satisfies~$\pi \cdot \vect{u} = \vect{v}$, or~$u_{\pi(i)} = v_i$
for each~$i \in [l]$.  Now we extend this equivalence relation~$\sim$
from the set ~$V^{l,\not=}$ to the set~$V^{l,\not=} \times
\powerset{\kpowerset{\geq l_1}{[l]}}$ as follows:~$(\vect{u},S) \sim
(\vect{v},T)$ if, and only if,~$\vect{u} \sim \vect{v}$ and the unique
permutation~$\pi \in S_l$ that sends~$\vect{u}$ to ~$\vect{v}$ also
sends~$S$ to~$T$; i.e., it satisfies~$\pi \cdot S = T$, where~$\pi
\cdot S$ denotes the natural action of~$\pi$ on ~$S$. It is not hard
to see that the set of equivalence classes
\begin{equation}
W'' := (V^{l,\not=} \times \powerset{\kpowerset{\geq l_1}{[l]}})/\!\sim
\label{eqn:altdefW}
\end{equation} 
is an alternative presentation of the same set~$W$.  This alternative
presentation of~$W$ is useful when we argue that the reduction is
an~$\FO$-interpretation in Theorem~\ref{thm:dinursafrareduction}
below.

We still need to define the vertex-weights and the edge-set
of~$H(G,\epsilon,p,h;r)$. The weight of a vertex~$(B,S)$ in~$W$ is
defined as
\begin{equation}
w(B,S) := M^q \cdot p^{|S|} \cdot (1-p)^{|\kpowerset{\geq
    l_1}{B}\setminus S|}, \label{eqn:weight}
\end{equation}
where~$M$ is the denominator of~$p = N/M$ reduced to lowest terms,
and~$q = |\kpowerset{\geq l_1}{B}|$. Next we define the edge-set: two
vertices~$(B_1,S_1)$ and~$(B_2,S_2)$ in~$W$ are adjacent if, and only
if, either~$B_1 = B_2$ and~$S_1 \cap S_2 = \emptyset$, or there exist
an edge~$\{v_1,v_2\} \in E$ of~$G$ and an~$(l-1)$-element
subset~$\hat{B}$ of~$V$ such that~$B_1 = \hat{B} \cup \{v_1\}$
and~$B_2 = \hat{B} \cup \{v_2\}$ and, for all~$(A_1,A_2) \in S_1
\times S_2$, either~$A_1 \cap \hat{B} \not= A_2 \cap \hat{B}$,
or~$(v_1,v_2) \in A_1 \times A_2$.


\begin{theorem}[Dinur-Safra Vertex-Cover Test \cite{DinurSafra}] \label{thm:dinursafralongcode}
For any two rationals~$\epsilon$ and~$p$ satisfying~$0 < \epsilon \leq
1$ and~$0 < p < p_{\mathrm{max}} = (3-\sqrt{5})/2$, any small
enough~$s_0 > 0$, any large enough integer~$h$, and
any~$(m,r)$-co-partite graph~$G$, the following hold:
\begin{enumerate} \itemsep=0pt
\item if~$\is{G} = m$, then~$\isd{H(G,\epsilon,p,h;r)} \geq p - \epsilon$,
\item if~$\his{h}{G} \leq s_0 m$, then~$\isd{H(G,\epsilon,p,h;r)} \leq
  p^{\bullet} + \epsilon$, where~$p^{\bullet} = \max(p^2,4p^3-3p^4)$.
\end{enumerate}
\end{theorem}

\noindent For the proof of Theorem~\ref{thm:dinursafralongcode}, see
Theorem~2.2 in~\cite{DinurSafra}.

The reduction described above produces a \emph{weighted}
graph~$H(G,\epsilon,p,h;r)$.  The weights, as defined
in~(\ref{eqn:weight}) are non-negative integers with a maximum value
of~$M^q$.  This value depends on~$\epsilon, p, h$ and~$r$ but is
independent of the number of vertices of~$G$.  In other words, fixing
the other parameters,~$H$ gives us a traslation from~$G$ to a weighted
graph, with integer weights bounded by a constant.  This can be easily
modified to get an unweighted graph.  Indeed, let~$H = (V,E,W)$ be a
graph with a weight function~$W: V \ra \nats$.  We define from this an
unweighted graph~$H'$ with~$\vcd{H'} = \vcd{H}$.  This is obtained by
replacing each vertex~$v$ by the set of vertices~$v^* := \{v\}\times
[W(v)]$ and having an edge between~$(u,i)$ and~$(v,j)$ if, and only
if,~$\{u,v\} \in E$.  To see that this has the right property, it is
sufficient to observe that~$S \subseteq V$ is a minimum weight vertex
cover in~$H$ if, and only if,~$S^* := \bigcup_{v \in S} v^*$ is a
minimum vertex cover of~$H'$.  The direction from right to left is
obvious.  For the other direction, suppose that~$H'$ has a minimum
vertex cover~$X$ that is not of this form.  In particular, for some~$v
\in V$,~$v^* \cap X \neq \emptyset$ and~$v^* \not\subseteq X$.  But
then~$X \setminus v^*$ is still a vertex cover, contradicting the
minimality of~$X$.

By composing Lemma~\ref{lem:bipartitereduction},
Theorem~\ref{thm:parallelrepetition} and
Theorem~\ref{thm:dinursafralongcode} with the appropriate parameters
and combining it with the observation above we get the following.

\begin{theorem} \label{thm:dinursafrareduction}
For every~$s,\epsilon,p$ with~$0 < s,\epsilon < 1$, and~$1/3 < p <
p_{\mathrm{max}} = (3-\sqrt{5})/2$, there is
an~$\FO$-interpretation~$\Theta$ that maps instances of~$\prob{3XOR}$
to undirected graphs in such a way that, for every~$\prob{3XOR}$
instance~$I$ the following hold:
\begin{enumerate} \itemsep=0pt
\item if~$I$ is satisfiable, then~$\vcd{\Theta(I)} \geq 1-p+\epsilon$,
\item if~$I$ is not~$s$-satisfiable, then~$\vcd{\Theta(I)} \leq
  1-p^{\bullet}-\epsilon$, where~$p^{\bullet} = \max(p^2,4p^3-3p^4)$.
\end{enumerate}
\end{theorem}

\begin{proof}
Firt we define~$\Theta(I)$ and then check that it is
an~$\FO$-interpretation. Let~$t$ be a large enough integer so that the
following inequality holds:
\begin{equation}
(1-(1-(s+2)/3)^3)^{\alpha t/6} \leq s_0/h^3
\end{equation}
when~$s_0$ is small enough, and~$h$ is large enough, so that
Theorem~\ref{thm:dinursafralongcode} applies. Such a~$t$ exists
because~$s < 1$ and~$s_0 > 0$. Apply the bipartite reduction to~$I$ to
obtain the instance~$I' = L(I)$ from
Lemma~\ref{lem:bipartitereduction}. The domain size of~$I'$
is~$6$. Apply the parallel repetition reduction of
Theorem~\ref{thm:parallelrepetition} to~$I'$ with parameter~$t$ to get
another instance~$I''$. Next apply the intermediate reduction of
Lemma~\ref{lem:intermediatereduction} to get a graph~$G$. Finally,
apply the Dinur-Safra long-code reduction of
Theorem~\ref{thm:dinursafralongcode} to get a weighted graph~$H$ and
convert it to an unweighted graph that is the output of~$\Theta$. The
parameters were chosen in such a way that the points 1 and 2 hold via
the relationship~$\vcd{H} = 1 - \isd{H}$.

We still need to check that~$\Theta$ is an~$\FO$-interpretation.  As
in the proof of Theorem~\ref{thm:gapamplireduction}, producing~$I'$
from~$I$ and~$I''$ from~$I'$ is straightforward. Producing~$G =
G(I'')$ from~$I''$ is equally straightforward: the definition of the
intermediate reduction is explicit enough that this can be checked
directly, especially because the weights of~$I''$ are still zeros and
ones. On the other hand, producing~$H = H(G,\epsilon,p,h; r)$ from~$G$
requires some explanation.

In the description of the vertex-cover long-code reduction we already
described~$W''$ as an alternative presentation~\eqref{eqn:altdefW}
of~$W$ in~\eqref{eqn:defW}. This alternative presentation suggests
that the vertex-set of~$H$ be defined by an~$\FO$-interpretation of
dimension~$l$ through the method of finite expansions from
Section~\ref{sec:preliminaries} to produce~$W'$ in~\eqref{eqn:proxy},
followed by a quotient by an~$\FO$-definable equivalence relation. The
method of finite expansions produces a set of the form~$V^{l,\not=}
\times A$ for some bounded set~$A$ that
codes~$\powerset{\kpowerset{l_1}{[l]}}$. The effect of the quotient
on~$V^{l,\not=} \times A$ can be achieved through the
equality-defining formula~$\epsilon(\vect{x},\vect{y})$ of
the~$\FO$-interpretation, which in this case can be designed as
follows. Let~$(\vect{u},\vect{a})$ and~$(\vect{v},\vect{b})$ be two
elements of the expanded domain~$V^{l,\not=} \times A$.  We
want~$\epsilon(\vect{x},\vect{y})$ to tell if~$\vect{u}$
and~$\vect{v}$ involve exactly the same elements from~$V$ and, in such
a case, whether the unique permutation that takes~$\vect{u}$
to~$\vect{v}$ also takes the set of subsets of~$[l]$ coded
by~$\vect{a}$ to the set of subsets of~$[l]$ coded by~$\vect{b}$. The
first part can be stated by means of a simple quantifier-free
formula. The second part can also be stated by a quantifier-free
formula (that depends on~$l$) by taking a disjunction over all~$l!$
potential permutations of~$[l]$.

Once the domain is defined as~$W''$ in equation~\eqref{eqn:altdefW},
defining the edge-set is easy. Defining the weights is also
straightforward given that~$h$,~$\epsilon$,~$p$ and~$l$ are all fixed
constants independent of~$G$, and as noted above, we can replace the
weights with sets of unweighted vertices.
\end{proof}

Now we can state the improved version of Corollary~\ref{cor:Ck-VC1}.
Composing Theorem~\ref{thm:3lin-onesided},
Theorem~\ref{thm:dinursafrareduction}, and Lemma~\ref{lem:reduction}
we get the following.

\begin{theorem}\label{thm:Ck-VC-improved}
For any~$\epsilon > 0$ there is a~$\delta > 0$ such that
if~$\class{C}$ is the collection of graphs~$G$ with~$\vcd{G} \leq
1-4p_{\mathrm{max}}^3 + 3p_{\mathrm{max}}^4-\epsilon$ and ~$\class{D}$
is the collection of graphs~$G$ with~$\vcd{G} \geq
1-p_{\mathrm{max}}+\epsilon$ then~$\class{C}$ and~$\class{D}$ are
not~$\Ck$-separable for any~$k = k(n)$ such that~$k(n) =
o(n^{\delta})$, where~$p_{\max} = (3-\sqrt{5})/2$.
\end{theorem}

\noindent 
In terms of algorithms, Theorem~\ref{thm:Ck-VC-improved} says that no
algorithm that can be expressed in~$\FPC$, or even~$\Ck$ for~$k =
n^{o(1)}$, can achieve an approximation ratio better
than~$(1-4p_{\max}^3 + 3p_{\max}^4)/(1-p_{\max}) \approx 1.36$.  In
particular, this means that~$n^{O(1)}$ levels of the Lasserre
hierarchy are necessary to give an approximation algorithm for vertex
cover with an approximation better than~$1.36$.  This result was
previously known from the work of Tulsiani~\cite{Tul09}.

\subsection{Tight lower and upper bounds for~$C^2$}

There are straightforward polynomial-time algorithms that yield a
vertex cover in a graph with guaranteed approximation ratio~$2$.  It
is conjectured that no polynomial-time algorithm can achieve an
approximation ratio of~$2-\epsilon$ for \emph{any}~$\epsilon > 0$; it
is even conjectured that achieving such an approximation ratio is
NP-hard.  It would be interesting to prove a version of this
conjecture for algorithms expressible in~$\FPC$, and without the
assumption that~$\PT\neq \NP$.  This could be established by a
strengthened version of Theorem~\ref{thm:Ck-VC-improved} with better
ratios.  We next show that we can at least do this for the special
case of~$k=2$.

\begin{theorem}\label{thm:c2vc-lower}
For any~$\epsilon > 0$, if~$\class{C}$ is the collection of graphs~$G$
with~$\vcd{G} \leq 1/2$ and~$\class{D}$ is the collection of
graphs~$G$ with~$\vcd{G} \geq 1-\epsilon$ then~$\class{C}$
and~$\class{D}$ are not~$C^2$-separable.
\end{theorem}
\begin{proof}
Let~$(G_n)_{n\in\NN}$ be a family of~$3$-regular expander graphs
on~$n$ vertices, so that the largest independent set in~$G_n$ has
size~$o(n)$. For the existence of such graphs see
\cite[Chapter~4]{VadhanSurvey}.  It follows that the smallest vertex
cover in~$G_n$ has size~$n - o(n)$.  Hence, we can choose a value ~$m$
such that~$G_{2m}$ has no vertex cover smaller than ~$2m(1-\epsilon)$.

Let~$H_m$ be a~$3$-regular bipartite graph on two sets of~$m$
vertices.  Now, each part of a bipartite graph is a vertex cover,
so~$H_m$ has a vertex cover of size~$m$.  However, it is known that~$G
\Cequiv{2} H$ holds for any pair~$G$ and~$H$ of~$d$-regular graphs
with the same number of vertices, for any~$d$.  Thus,~$G_{2m}
\Cequiv{2} H_m$ and the result follows.
\end{proof}

Essentially, Theorem~\ref{thm:c2vc-lower} tells us that no algorithm
that is invariant under~$\Cequiv{2}$ can determine~$\vc{G}$ to an
approximation better than~$2$, and Theorem~\ref{thm:Ck-VC-improved}
tells us that no algorithm that is invariant under~$\Cequiv{k}$ for
constant or even slowly growing~$k$ can determine~$\vc{G}$ to an
approximation better than~$1.36$.  A legitimate question at this point
is whether there is any algorithm that is invariant
under~$\Cequiv{k}$, such as one expressible in~$\FPC$ would be, that
\emph{does} achieve an approximation ratio of~$2$.  The natural
polynomial-time algorithms that give a vertex cover with size at
most~$2\vc{G}$ are not expressible in~$\FPC$.  Indeed, we cannot
expect a formula of~$\FPC$ to define an actual vertex cover in a
graph~$G$ as this is not invariant under automorphisms of~$G$.  We can
only ask for an estimate of the size, i.e.\ of~$\vc{G}$, and this we
can get up to a factor of~$2$.  For this, it turns out that~$k=2$ is
enough, showing that the lower bound of Theorem~\ref{thm:c2vc-lower}
is tight:

\begin{theorem}\label{thm:c2vc-upper}
For any~$\delta$, if~$\class{C}$ is the collection of graphs~$G$
with~$\vcd{G} \leq \delta$ and~$\class{D}$ is the collection of
graphs~$G$ with~$\vcd{G} > 2\delta$ then~$\class{C}$ and ~$\class{D}$
are~$\Cequiv{2}$-separable.
\end{theorem}

The proof of Theorem~\ref{thm:c2vc-upper} proceeds through a series of
lemmas.

\begin{lemma}\label{lem:half-regular}
If~$G$ is a~$d$-regular graph on~$n$ vertices, for any~$d \geq 1$,
then~$\vc{G} \geq n/2$.
\end{lemma}
\begin{proof}
Let~$S$ be any set of vertices in~$G$.  Then the number of edges
incident on vertices in~$S$ is at most~$d|S|$.  Since the number of
edges in~$G$ is~$d n/2$, if~$S$ is a vertex cover~$d |S| \geq d n/2$
and so~$|S| \geq n/2$.
\end{proof}

Let~$G$ be a graph and~$C_1,\ldots,C_m$ be the partition of the
vertices of~$G$ given by \emph{vertex refinement}.  So, there are
constants~$\delta_{ij}$ such that each~$v \in C_i$ has
exactly~$\delta_{ij}$ neighbours in~$C_j$.  Since the graph is
undirected, the number of edges from~$C_i$ to~$C_j$ is the same as in
the other direction and so~$\delta_{ij}|C_i| = \delta_{ji}|C_j|$, for
all~$i$ and~$j$.  Also,~$\delta_{ij} = 0$ if, and only
if,~$\delta_{ji} = 0$.

Let~$X = \{ i \mid \delta_{ii} = 0\}$ and~$Y = \{i \mid \delta_{ii} >
0 \}$.  Consider the undirected graph~$X_G$ with vertices~$X$ and
edges~$\{ (i,j) \mid \delta_{ij} > 0\}$.  Consider the
instance~$(X_G,w)$ of \emph{weighted vertex cover} obtained by taking
the graph~$X_G$ and giving each vertex~$i$ the weight~$w(i) = |C_i|$.
Let~$p_G$ denote the value of the minimum weighted vertex cover of
this instance.  Also, let~$q_G = \sum_{i\in Y} |C_i|$.  Finally,
define~$v_G = p_G + q_G$.

\begin{lemma}\label{lem:c2vc-invariant}
If~$G \Cequiv{2} H$ then~$v_G = v_H$.
\end{lemma}
\begin{proof}
The value~$v_G$ is determined entirely by the sizes of~$C_i$ in the
vertex refinement of~$G$ and the corresponding values
of~$\delta_{ij}$.  Since~$G \Cequiv{2} H$, these values are the same
for~$H$.
\end{proof}

\begin{lemma}\label{lem:c2vc-lower}~$\vc{G} \leq v_G$.
\end{lemma}
\begin{proof}
Let~$Z \subseteq X$ be a minimum-weight vertex cover in~$(X_G,w)$.
Take the set~$S \subseteq V(G)$ defined by ~$S = \bigcup_{i \in Y \cup
  Z} C_i$.  Note that the sets~$Y$ and~$Z$ are disjoint, ~$\sum_{i\in
  Y} |C_i| = q_G$ by definition, and~$\sum_{i\in Z} |C_i| = p_G$ by
construction.  So~$S$ has exactly~$v_G$ vertices.  We claim that~$S$
is a vertex cover in~$G$.  Let~$e$ be any edge of~$G$ with endpoints
in~$C_i$ and~$C_j$.  If either~$i$ or~$j$ is in~$Y$, then the
corresponding endpoint of~$e$ is in~$S$ since~$C_i \subseteq S$ for
all~$i\in Y$.  If both~$i$ and~$j$ are not in~$Y$ then both are in~$X$
and~$\delta_{ij} > 0$.  Thus, since~$Z$ is a vertex cover for the
graph~$X_G$ then one of~$i$ or~$j$ must be in~$Z$ and again at least
one endpoint of~$e$ is in~$S$.
\end{proof}

For the proof of the next lemma, we need the notion of a
\emph{fractional vertex cover} of a graph~$G=(V,E)$.  This is a
function~$f: V \ra [0,1]$ satisfying the condition that for
every~$(u,v) \in E$,~$f(u) + f(v) \geq 1$.  It is known that if~$f$ is
a fractional vertex cover of~$G$, then~$\sum_{v \in V} f(v) \geq
\vc{G}/2$ (see~\cite[Thm.~14.2]{Vaz03}).  More generally, suppose we
have an instance of \emph{weighted vertex cover}, i.e.~$G$ along with
a weight function~$w: V \ra \nats$ where~$\vc{G,w}$ is defined as the
value of the minimum weighted vertex cover.  Then~$\sum_{v \in V}
f(v)w(v) \geq \vc{G,w}/2$.

\begin{lemma}\label{lem:c2vc-upper}~$v_G \leq 2\vc{G}$.
\end{lemma}
\begin{proof}
Let~$S$ be any vertex cover of~$G$.  Let~$U_X = \bigcup_{i \in X} C_i$
and~$U_Y = \bigcup_{i \in Y} C_i$ and note that these sets are
disjoint.  We claim that~$|S \cap U_X| \geq p_G/2$ and~$|S \cap U_Y|
\geq q_G/2$, and therefore~$|S| = |S \cap U_X| + |S \cap U_Y| \geq
v_G/2$, establishing the result.

First, consider~$S \cap U_Y$.  Note that for any~$i \in Y$, the
subgraph of~$G$ induced by~$C_i$ is~$\delta_{ii}$-regular.
Since~$\delta_{ii} > 0$ by definition of~$Y$, by
Lemma~\ref{lem:half-regular} we have~$|S \cap C_i| \geq |C_i|/2$ and
therefore~$|S \cap U_Y| \geq q_G/2$.

Secondly, consider the function~$f: X \ra [0,1]$ defined by~$f(i) = |S
\cap C_i|/|C_i|$.  We claim that this is a fractional vertex cover of
the graph~$X_G$.  To verify this, we need to check that ~$f(i) + f(j)
\geq 1$ whenever~$\delta_{ij} > 0$.  There are ~$\delta_{ij}|C_i|$
edges between~$C_i$ and~$C_j$.  Each element of ~$S\cap C_i$ can cover
at most~$\delta_{ij}$ of these edges and similarly each element
of~$S\cap C_j$ covers at most~$\delta_{ji}$ of them.  Thus, since~$S$
is a vertex cover~$| S \cap C_i| \delta_{ij} + |S \cap C_j|
\delta_{ji} \geq \delta_{ij}|C_i|$.  Substituting for~$\delta_{ji}$
using the identity~$\delta_{ij}|C_i| = \delta_{ji}|C_j|$ gives~$| S
\cap C_i| \delta_{ij} + |S \cap C_j| \delta_{ij} |C_i|/|C_j| \geq
\delta_{ij}|C_i|$. Now dividing through by~$\delta_{ij}|C_i|$
gives~$f(i) + f(j) \geq 1.$

Thus, we have that the weighted vertex cover instance~$(X_g,w)$ admits
the fractional solution~$f$ whose total weight is ~$$\sum_{i \in X}
f(i)|C_i| = \sum_{i\in X} |S \cap C_i| = |S \cap U_X|.$$ Since~$p_G$
is the value of the minimum weight vertex cover of~$(X_g,w)$, we
have~$ |S \cap U_X| \geq p_G/2$, as was to be shown.
\end{proof}

\begin{proof}[Proof of Theorem~\ref{thm:c2vc-upper}]
Suppose for contradiction that there is a~$G \in \class{C}$ and~$H \in
\class{D}$ such that~$G \Cequiv {2} H$.  Since~$G$ and~$H$ must have
the same number of vertices, we have~$2 \vc{G} < \vc{H}$.  But, by
Lemma~\ref{lem:c2vc-upper} we have~$v_G \leq 2\vc{G}$, by
Lemma~\ref{lem:c2vc-lower} we have~$\vc{H} \leq v_H$ and by
Lemma~\ref{lem:c2vc-invariant} we have~$v_G = v_H$, giving a
contradiction.
\end{proof}

\section{Conclusions}
This paper introduces a new method for studying the hardness of
approximability of~$\NP$-hard optimization problems by showing that
the approximation cannot be \emph{defined} in a suitable logic such
as~$\FPC$.  This is done by showing that no class of bounded counting
width can separate instances of the problem with a high optimum from
those with a low one.  This raises a large number of new challenges in
the application of this method.  A clear demonstration of the power of
this method would be to derive a lower bound stronger than one for
which~$\NP$-hardness is known. For instance, can we improve, in the
context of inexpressibility, on the~$\sqrt{2}$-inapproximability for
vertex cover from the~$\NP$-hardness result of Khot et
al.~\cite{KhotMinzerSafra}?  In other words, can we show that the
class of graphs that have a vertex cover of density~$\delta$ is not
separable from the class of graphs that do not have a vertex cover of
density~$c\delta$, for some~$\delta \in (0,1)$ and some constant~$c$
greater than~$\sqrt{2}$? If this were achieved for~$\Ck$, for
unbounded~$k$, it would have major consequences in the study of
semidefinite programming hierarchies of relaxations of vertex cover.
A version of this question, with~$\delta$ being~$1-\epsilon$ and~$c$
being~$2-\epsilon$ for arbitrary small~$\epsilon$, was stated as Open
Problem~4.1 in~\cite{FAMT-Dagstuhl}.  Indeed, similar questions can be
posed for any optimization problem for which the exact
inapproximability factor is not known, including~$\prob{MAX CUT}$,
sparsest cut, etc.

\bigskip \bigskip \noindent\textbf{Acknowledgments.}  The research
reported here was initiated at the Simons Institute for the Theory of
Computing during the programme on Logical Structures in Computation in
autumn 2016.  The first author was partially funded by European
Research Council (ERC) under the European Union's Horizon 2020
research and innovation programme, grant agreement ERC-2014-CoG 648276
(AUTAR) and MICCIN grant TIN2016-76573-C2-1P (TASSAT3).  The second
author was partially supported by a Fellowship of the Alan Turing
Institute under the EPSRC grant EP/N510129/1 and by the EPSRC grant
EP/S03238X/1

\bibliographystyle{plain}
\bibliography{indistinguishable}

\end{document}